\newif\iflncs
\newif\ifsubmission 
\newif\ifdisplayextendedabstract 
\newif\ifshowrevisions 
\newif\ifshowcomments 

\lncsfalse
\submissionfalse
\displayextendedabstractfalse
\showrevisionsfalse
\showcommentsfalse

\ifsubmission
  \showrevisionsfalse
  \showcommentsfalse
\fi

\iflncs
  \documentclass{llncs}
\else
  \documentclass{article}
\fi

\iflncs
  \usepackage{amsmath,amssymb,mathtools}
\else
  \usepackage[margin=1in]{geometry}
  \usepackage[T1]{fontenc}
  \usepackage{amsmath,amssymb,amsthm,mathtools}
\fi
\usepackage{graphicx}

\iflncs
\else
  \usepackage[numbers,sort&compress]{natbib}
\fi

\usepackage{algorithm}
\usepackage{algpseudocode}
\usepackage{mdframed}
\usepackage{xcolor}
\usepackage{xurl}
\usepackage{hyperref}
\usepackage{cleveref}
\hypersetup{
  colorlinks=true,
  linkcolor=red,    
  citecolor=blue
}
\iflncs
  \hypersetup{hidelinks}
\fi
\ifsubmission
  \hypersetup{pdfauthor={},pdfsubject={},pdfkeywords={}}
\fi

\makeatletter
\renewcommand{\theHALG@line}{\thealgorithm.\arabic{ALG@line}}
\makeatother

\iflncs
  \spnewtheorem{fact}{Fact}{\bfseries}{\itshape}
  \let\LNCSproof\proof
  \let\endLNCSproof\endproof
  \RenewDocumentEnvironment{proof}{o}
    {\IfNoValueTF{#1}{\LNCSproof}{\renewcommand{\proofname}{#1}\LNCSproof}}
    {\qed\endLNCSproof}
\else
  \newtheorem{theorem}{Theorem}[section]
  \newtheorem{lemma}[theorem]{Lemma}

  \newtheorem{corollary}[theorem]{Corollary}
  \theoremstyle{definition}
  \newtheorem{definition}[theorem]{Definition}
  \newtheorem{remark}[theorem]{Remark}
  \AddToHook{env/lemma/begin}{\crefalias{theorem}{lemma}}
  \AddToHook{env/proposition/begin}{\crefalias{theorem}{proposition}}
  \AddToHook{env/fact/begin}{\crefalias{theorem}{fact}}
  \AddToHook{env/corollary/begin}{\crefalias{theorem}{corollary}}
  \AddToHook{env/definition/begin}{\crefalias{theorem}{definition}}
  \AddToHook{env/remark/begin}{\crefalias{theorem}{remark}}
\fi

\crefname{fact}{fact}{facts}
\Crefname{fact}{Fact}{Facts}

\newcommand{\ket}[1]{\lvert #1\rangle}
\newcommand{\bra}[1]{\langle #1\rvert}
\newcommand{\braket}[2]{\langle #1\mid #2\rangle}
\newcommand{\proj}[1]{\lvert #1\rangle\!\langle #1\rvert}
\newcommand{\abs}[1]{\left\lvert #1\right\rvert}
\newcommand{\norm}[1]{\left\lVert #1\right\rVert}
\newcommand{\E}{\mathbb{E}}
\newcommand{\C}{\mathbb{C}}
\newcommand{\calH}{\mathcal{H}}
\newcommand{\cA}{\mathcal A}
\newcommand{\cC}{\mathcal C}
\newcommand{\cD}{\mathcal D}
\newcommand{\cF}{\mathcal F}

\newcommand{\cL}{\mathcal L}
\newcommand{\cM}{\mathcal M}
\newcommand{\cO}{\mathcal O}

\newcommand{\Id}{I}
\newcommand{\dd}{\mathrm{d}}
\DeclareMathOperator{\Tr}{Tr}
\DeclareMathOperator{\Sym}{Sym}
\DeclareMathOperator{\negl}{negl}
\DeclareMathOperator{\poly}{poly}
\DeclareMathOperator{\im}{im}
\newcommand{\Bank}{\mathsf{Bank}}
\newcommand{\Ver}{\mathsf{Ver}}
\newcommand{\pk}{\mathsf{pk}}
\newcommand{\sk}{\mathsf{sk}}
\newcommand{\Adv}{\cA}
\newcommand{\eps}{\varepsilon}
\newcommand{\supp}{\operatorname{supp}}
\newcommand{\TI}{\mathsf{TI}}
\newcommand{\API}{\mathsf{API}}
\newcommand{\ATI}{\mathsf{ATI}}

\newcommand{\dum}{\mathsf{dum}}

\newcommand{\acc}{\mathsf{acc}}
\newcommand{\HClone}{\mathsf{HClone}}
\newcommand{\Pir}{\mathsf{CopyProtect}}
\newcommand{\Learn}{\mathsf{Learn}}
\newcommand{\Setup}{\mathsf{Setup}}
\newcommand{\Gen}{\mathsf{Generate}}
\newcommand{\Eval}{\mathsf{Eval}}

\newcommand{\one}[1]{\mathbf 1\!\left[#1\right]}
\newcommand{\learn}{\mathrm{learn}}
\newcommand{\clone}{\mathrm{clone}}

\newcommand{\col}{\mathsf{col}}
\newcommand{\negg}{\mathsf{neg}}
\newcommand{\bal}{\mathsf{bal}}

\NewDocumentCommand{\revised}{+g}{\IfNoValueTF{#1}
    {\ifshowrevisions\color{brown}\fi}{\ifshowrevisions{\color{brown}#1}\else#1\fi}}

\title{More is Less:\\
Optimal Security for Haar Quantum Money and More}

\author{Zihan Hao\thanks{z1hao@ucsd.edu}  \\ \small{UC San Diego} \and Xingjian Li\thanks{lxj22@mails.tsinghua.edu.cn}  \\ \small{Tsinghua University} \and Qipeng Liu\thanks{qipengliu0@gmail.com} \\ \small{UC San Diego} \and Wei Zhan\thanks{weizhan@purdue.edu} \\ \small{Purdue University}}
\date{}

\iflncs
  \author{}
  \institute{}
  \authorrunning{}
  \titlerunning{}
\fi

\ifsubmission
  \author{}
\fi

\begin{document}

\maketitle

\thispagestyle{plain}
\pagestyle{plain}

\begin{abstract}
Quantum cryptography leverages unclonability to enable a wide range of
cryptographic applications that are impossible classically, such as digital
currency protected against counterfeiting by quantum mechanics. Security
requires that no efficient user can produce even one additional valid
banknote beyond those already in their possession. However, existing
security bounds weaken as the number of banknotes available to a user
increases.

In this paper, we study a construction of quantum money
whose asymptotic query security does not deteriorate as long as the number
of banknotes available to a user remains below the scale required for state
tomography. In particular, we show that the construction based on an
$n$-qubit Haar-random state and a reflection oracle achieves optimal query
security: unless a user holds $\Omega(2^n)$ banknotes, their existing banknotes
cannot asymptotically speed up counterfeiting --- their best possible attack is the same as if they do not have any banknotes. To establish this result, we develop a framework based on the
compressed-oracle technique for defining and analyzing progress measures
for quantum tasks such as Haar state cloning. Furthermore, we prove a tight
lower bound for generating $r$ additional copies, give a matching attack,
and apply our results to quantum copy-protection.
\end{abstract}

\newpage

\section{Introduction}

Quantum states enable cryptographic functionalities that have no classical
counterpart. A foundational example is quantum money, first proposed by
Wiesner~\cite{wiesnerConjugateCoding1983}: a bank issues quantum banknotes
that can be verified and spent, but cannot be efficiently counterfeited.
Security requires that an efficient adversary given $k$ honestly minted
banknotes cannot produce $k+1$ valid banknotes, except with negligible
probability. Public-key quantum money further allows anyone to verify a
banknote without knowing the bank's secret
key~\cite{aaronsonQuantumCopyProtectionQuantum2009,aaronsonQuantumMoneyHidden2012}.
A particularly natural and stronger form of quantum money is \emph{quantum coins}, in
which all honestly minted coins are identical
copies of a single pure
state~\cite{moscaQuantumCoins2010,beheraAlmostPublicQuantumCoins2020}.
Security must therefore hold even when an adversary possesses many copies
of the same state.

Ideally, the cost of producing one additional banknote should not decrease as
an adversary acquires more banknotes.
If this deterioration reflects a genuine advantage, it would give wealth
an unintended role: possessing more money would make it easier to attack the system. The concern is not simply that a
wealthier adversary can afford more computational resources (which they always can), but that
the banknotes themselves might provide a shortcut to counterfeiting. This
would give the familiar phrase ``the rich get richer'' an unwelcome
cryptographic meaning.
A monetary system need not promise equal wealth, but it should not turn
wealth into an advantage in evading its rules. The principle we seek is
therefore simple: greater wealth should not grant greater counterfeiting power.

This property is not automatic. For
example, measuring $O(n)$ identical copies of an $n$-qubit hidden-subspace
banknote in the computational basis reveals a spanning set for the secret
subspace with high probability, allowing the adversary to prepare further
copies~\cite[Section~5.1]{aaronsonQuantumMoneyHidden2012}. For Haar-state
constructions with reflection-oracle verification, previous no-cloning
bounds also weaken as the number $k$ of input copies increases. Aaronson's
complexity-theoretic no-cloning theorem establishes security with
polynomially many input copies, but its quantitative bound depends on
$k$~\cite{aaronsonQuantumCopyProtectionQuantum2009}. Aaronson and Christiano
subsequently proved an $\Omega(\sqrt{2^n/k})$ query lower bound for exact
$k\to k+1$ cloning with certainty in the worst
case~\cite[Appendix~10, Theorem~43]{aaronsonQuantumMoneyHidden2012}.
These bounds leave open whether additional copies actually reduce the
query cost of counterfeiting. Improving the parameters of the
complexity-theoretic no-cloning theorem was explicitly posed as an open
problem by
Aaronson~\cite[Section~5]{aaronsonQuantumCopyProtectionQuantum2009}.

In this work, we show that the Haar-state reflection-oracle construction
underlying these proposals~\cite{aaronsonQuantumCopyProtectionQuantum2009,moscaQuantumCoins2010}
has the desired property: as long as $k=o(2^n)$, producing one additional
valid banknote with constant probability requires $\Theta(2^{n/2})$ queries,
the same asymptotic cost as preparing a banknote with no input copies. Thus,
throughout this regime, possessing more banknotes provides no asymptotic
advantage in counterfeiting query complexity.

\subsection{Our Results}

We first state our bound for producing one additional copy. The
algorithm receives $k$ copies of a state $\ket{\psi}$ and query access
to the controlled reflection
\[
  C_{\ket{\psi}}Z
  = I\otimes(I-\proj{\psi})+Z\otimes\proj{\psi}
  = I-2\proj{1}\otimes\proj{\psi}.
\]
This oracle implements the projective measurement
$\{\proj{\psi},I-\proj{\psi}\}$ with one query. The algorithm succeeds
if all $k+1$ output registers pass the projection onto $\ket{\psi}$; its
output may be entangled. We count oracle queries and allow arbitrary
oracle-independent operations.

\begin{theorem}[Tight bounds for $k\to k+1$ cloning]
\label{thm:intro-k-to-k-plus-one}
Let $N=2^n$, and let $\ket{\psi}$ be a Haar-random state in $\C^N$.
Any algorithm given $k$ copies of $\ket{\psi}$ that makes at most $T$
queries to $C_{\ket{\psi}}Z$ and outputs a state $\rho_\psi$ on $k+1$
registers satisfies
\[
  \E_{\ket{\psi}}
  \Tr\!\left[(\proj{\psi})^{\otimes(k+1)}\rho_\psi\right]
  \leq
  O\left(\frac{T^2}{N}+\frac{k+1}{N}\right),
\]
where $\rho_\psi$ includes the algorithm's internal randomness.
Consequently, when $k=o(N)$, the query complexity of producing one
additional copy with probability at least $2/3$, averaged over the
Haar-random state, is $\Theta(\sqrt N)$. Moreover, the bound is asymptotically tight.
\end{theorem}

The lower bound is proved in \Cref{thm:haar-k-to-k-plus-one}. The
constant-success query bound is matched by the one-copy case of
\Cref{thm:0-to-k}. Indeed, we prove \Cref{thm:intro-k-to-k-plus-one} (as well as \Cref{thm:intro-k-to-k-plus-r}) for pseudorandom states based on random functions~\cite{JLS18} and then extend the bounds to Haar-random states by random self-reduction and the left invariance of Haar measure; thus, the main message conveyed by both theorems applies to collections of quantum states with classical descriptions.

\begin{remark}
Given $k$ copies of a Haar-random pure state in $\C^N$, the optimal
expected squared overlap between the estimated state and the true state
is $(k+1)/(k+N)$~\cite{brussOptimalStateEstimation1999}.
The optimal universal $k\to k+1$ cloner achieves the same joint
acceptance probability without any oracle
queries~\cite{wernerOptimalCloningPureStates1998}. Thus, $k=\Theta(N)$
is the threshold for 
constant-success cloning without oracle queries. Our first theorem says that $k=\Theta(N)$ is still the threshold for constant-success cloning even with oracle access.
\end{remark}

\medskip

Our lower-bound framework extends to producing several additional copies.
We also give an algorithm with matching query complexity, obtaining the
following generalization to $k\to k+r$ cloning.

\begin{theorem}[Tight bounds for $k\to k+r$ cloning]
\label{thm:intro-k-to-k-plus-r}
Let $N=2^n\geq2$, and let $k,r\geq1$ be integers with $r\leq N$.
For a Haar-random state $\ket{\psi}\in\C^N$, any algorithm given $k$
copies of $\ket{\psi}$ that makes at most $T$ queries to $C_{\ket{\psi}}Z$
and outputs a state $\rho_\psi$ on $k+r$ registers satisfies
\[
  \E_{\ket{\psi}}
  \Tr\!\left[(\proj{\psi})^{\otimes(k+r)}\rho_\psi\right]
  \leq
  \left(
    O\left(\frac{T^2}{{rN}}+{\frac{k+r}{N}}
  \right)\right)^{r}.
\]
The output $\rho_\psi$ may be entangled
and includes the algorithm's internal randomness.

Conversely, there is an algorithm that, for every state
$\ket{\psi}\in\C^N$, takes $k$ input copies, makes $O(\sqrt{rN})$
queries to $C_{\ket{\psi}}Z$, and outputs exactly $\ket{\psi}^{\otimes(k+r)}$
with probability at least $2/3$. Consequently, when $k+r=o(N)$, the
query complexity of producing $r$ additional copies with constant success
probability, averaged over the Haar-random state, is $\Theta(\sqrt{rN})$.
\end{theorem}

The lower bound is proved in \Cref{thm:haar-k-to-k-plus-r}, and the
matching upper bound is proved in \Cref{cor:k-to-k-plus-r-upper}.

\begin{remark}
The dependence on $r$ reveals a further feature of state generation.
When $k+r=o(N)$, producing $r$ additional copies together with constant
joint success probability has optimal query complexity $\Theta(\sqrt{rN})$,
or $\Theta(\sqrt{N/r})$ queries per new copy on average.
In other words, repeatedly applying the optimal algorithm from \Cref{thm:intro-k-to-k-plus-one} does not yield the optimal algorithm for the case of $k \to k+r$. Thus, the task benefits from batch generation: the more copies produced, the lower the average cost per copy.
Thus, producing more copies together reduces the average cost per new
copy, while possessing more copies in advance gives no asymptotic
advantage for the same cloning task.
\end{remark}

\paragraph{Cryptographic applications.}

Our cloning bounds yield two cryptographic applications. First, they
immediately give optimal query security for Haar quantum money in the
reflection-oracle model~\cite{aaronsonQuantumCopyProtectionQuantum2009,moscaQuantumCoins2010}. Since all honest
banknotes are copies of one Haar-random state and verification projects onto
that state, producing $r$ additional valid banknotes is precisely the
$k\to k+r$ cloning task. Thus, when $k+r=o(N)$, counterfeiting $r$
additional banknotes with constant probability requires $\Theta(\sqrt{rN})$
queries, independently of the number $k$ of banknotes already held. In particular, for the most interesting case $r=1$, it gives quantum money the ideal property: greater wealth should not confer greater counterfeiting power.

\begin{theorem}[Tight security of quantum money, informal]
\label{thm:intro-quantum-money}
Let $N=2^n$ where $n$ is the number of qubits of each banknote. There exists a quantum money scheme relative to an oracle such that any counterfeiter given $k\geq1$ valid banknotes and making at most $T$
oracle queries produces $k+1$ banknotes that all pass verification
with probability at most
\[
  O\left(\frac{T^2+{k+1}}{N}\right).
\]
Consequently, when $k=o(N)$, counterfeiting one additional banknote
with success probability at least $2/3$ has query complexity
$\Theta(\sqrt N)=\Theta(2^{n/2})$, independently of $k$.
\end{theorem}

Second, we adapt the copy-protection approach of Aaronson, Liu, Liu,
Zhandry, and Zhang~\cite{aaronsonNewApproachesQuantum2021}, which uses
hidden-subspace states, to Haar-random program states with a quantum oracle. Our construction works for all unlearnable function families whose verification does not require oracle calls to the original function. We specify this property later; many cryptographic applications, such as signing and decryption algorithms (both studied in~\cite{aaronsonNewApproachesQuantum2021}) satisfy this property.

\begin{theorem}[Security of our copy-protection scheme, informal]
\label{thm:intro-copy-protection}
Let $N=2^n$. There is a copy-protection construction for all function families with oracle-free verification in the
quantum-oracle model using $n$-qubit Haar random states with
the following guarantee. 
Any pirate given $k$ identical program copies and making $T_{\cA}$ evaluation-oracle queries satisfies, for every $0<\eta<1$,
\[
  \Pr[\text{piracy succeeds}]
  \leq \mu(T_\cA)+\HClone_{N,k}(T^*)+\eta,
  \qquad T^*=T_{\cA}+P(k,n,1/\eta).
\]
Success means producing $k+1$ programs that jointly pass the
threshold correctness tests. Here $\mu(q)$ is the maximum
success probability of a $q$-query learner for the underlying function
family at the specified correctness threshold, under the same input
distribution; $\HClone_{N,k}(q)$ is the optimal success probability
for $k\to k+1$ Haar-state cloning using $q$ reflection-oracle queries.
The polynomial $P$ depends only on the fixed
correctness parameters and output-program query bound, and is independent of $T_{\cA}$.
\end{theorem}

The three terms account for learning the protected function, cloning
the program state, and the reduction's approximation error,
respectively~(\Cref{thm:copy-protect}). Writing
$P=P(k,n,1/\eta)$, our cloning bound gives
\[
  \HClone_{N,k}(T^*)
  \leq O\!\left(\frac{T_{\cA}^{2}+P^{2}+k+1}{N}\right).
\]
In particular, the contribution $T_{\cA}^{2}/N$ has no multiplicative
loss depending on $k$. 
For $k$ below a suitable threshold $k^*$, the term $\negl := ({P^2+k+1})/{N}$ is negligible. This gives the following asymptotic bound for breaking the underlying copy-protection protocol:
\[
  \mu(T_\cA)+\frac{T_{\cA}^2}{N}+\eta+\negl.
\]
This is the sense in which many copies do not lower the
cloning barrier: one must either break the learning game (w.r.t. $\mu(T_\cA)$) or the cloning game (w.r.t. $\frac{T_\cA^2}{N}$).

\begin{remark}
\c{C}akan and Goyal~\cite{cakanUnclonableCryptographyUnbounded2024}
already established unbounded collusion-resistant copy-protection for
unlearnable functionalities relative to a classical oracle. Their
construction issues freshly randomized programs, which can be transformed into identical pure-state copies using the constructions in~\cite{ananthLessMoreCopy2026,cakanMultiCopySecurity2025}; we analyze identical
pure-state copies in a quantum-oracle construction. Our focus is whether security weakens (i.e., whether the attacker's query complexity decreases) when the
number of copy-protection programs available to a pirate increases. 
\end{remark}

\medskip

Finally, our proofs of \Cref{thm:intro-k-to-k-plus-one,thm:intro-k-to-k-plus-r}
also establish the cloning bounds for phase states in the QROM, so the
above cryptographic applications extend to phase states as well.
Because phase states are specified by classical functions, they suggest
a possible route toward efficient plain-model constructions or heuristic
instantiations. Developing such instantiations and establishing their
security remain open directions.

\subsection{Technical Overview}

\subsubsection{Lower Bounds}
\label{sec:technical-overview-lower-bounds}

We sketch the proof idea of the lower bounds on generating extra copies.
We mainly focus on the more intuitive $k\to k+1$ case, while briefly mentioning how to generalize our methods to analyze the $k\to k+r$ case for arbitrary $r>0$. 

The proof consists of two steps: (1) use random self-reduction to switch to phase states in the quantum random oracle model (QROM), and (2) derive a framework for analyzing the cloning task for phase states in this model. Previously, almost all applications of random oracle techniques addressed classical problems, such as post-quantum security or problems with classical outputs; we provide one of the few cases where the random oracle model quantifies the advantage in an inherently quantum task, namely cloning phase states in the QROM.

\paragraph{Recap of the compressed oracle technique \cite{zhandry19}.}
We start by recalling the compressed oracle technique introduced by
Zhandry in \cite{zhandry19} and how it could be used to prove average-case lower
bounds on problems such as collision finding. In this setting, an algorithm has query access to a random oracle $f:[N]\to[N]$ and aims to output a pair of inputs $(x_1,x_2)$ such that $f(x_1)=f(x_2)$. The analysis follows this framework:

\begin{enumerate}
    \item \textbf{Purify the oracle.}
    The algorithm is given a random oracle $f$, which is equivalent to
    initializing an oracle register with the superposition
    $\frac{1}{\sqrt{N^N}}\sum_f\ket{f}$, while the oracle register controls
    the algorithm's queries. Further, viewing the oracle register in the
    Fourier basis, it begins with the all-zero table
    $\ket{\widehat{0}}^{\otimes N}$, and each query modifies one entry of
    the table. Only queried positions of the table are non-zero.

    \item \textbf{Define progress.}
    Let $\Gamma$ project onto the states with a collision among the queried
    positions. This approximates the winning condition, up to a
    $1/N$ probability of finding a collision by pure chance. We use the
    weight of the current state on $\Gamma$ to evaluate the progress.
    Since $\Gamma$ acts only on the oracle register, this weight is
    invariant under local operations.

    \item \textbf{Bound the increment.}
    One can first show that the chance of falling into $\Gamma$ from
    $I-\Gamma$ after a single query is small. Then, with $T$ queries, by
    the triangle inequality on amplitudes, the progress on $\Gamma$ is at most $T$ times the
    progress from one query.
\end{enumerate}

We borrow this framework and customize it to our problem with a
state reflection oracle and a Haar random distribution.

\paragraph{Reduce to a discrete distribution.}
If we directly apply the above lower-bounding framework, we immediately
face a challenge when purifying the oracle. In our problem, the Haar-random
distribution is continuous and needs infinite space to purify. To simplify
the analysis, we first reduce the problem to a discrete distribution
$\mathcal{D}_n$ by a random self-reduction argument.

For a state $\ket{\psi}\sim\mathcal{D}_n$, we sample a random oracle
$f:[N]\to[M]$ and encode $f$ into the phases,
$\ket{\psi}=\frac{1}{\sqrt{N}}\sum_x\omega^{f(x)}\ket{x}$, where
$\omega=e^{2\pi i/M}$. With such a distribution, we can label the oracle
by $f$ and therefore purify it using a superposition in a standard way.
One may notice that $\mathcal{D}_n$ and the Haar measure are already
statistically close with polynomial copies and queries, as shown in \cite{JLS18}; here, we can make the
reduction without any loss, through the following random self-reduction trick.

The random self-reduction can potentially reduce the problem to any
distribution $\mathcal{D}$, by the left invariance of the Haar measure.
If we have an algorithm $\mathcal{A}$ that performs well on a Haar-random
input state $\ket{\psi}$, then we can create an algorithm $\mathcal{A}'$
that first randomly twirls the oracle state, runs $\mathcal{A}$, and then
undoes the twirl on the output. Specifically, it samples a Haar-random
unitary $V$ and applies it to each state copy and the reflection oracle, so it is
effectively running $\mathcal{A}$ on the state $V\ket{\psi}$, which
follows the Haar-random distribution. Therefore, the lower bound for
$\mathcal{D}$ works for the Haar-random distribution as well. We can also replace the controlled reflection $C_{\ket{\psi}}Z$ with the plain reflection
$R_\psi=I-2\ket{\psi}\bra{\psi}$. More details can be found in
the proof of \Cref{thm:haar-k-to-k-plus-r}.

\paragraph{Purify the oracle.}
Now that each oracle state takes the form $\ket{\psi_f}$, with $f:[N]\to[M]$ being a random oracle, we can purify it in the standard way. Then each oracle query becomes a joint unitary $R$, with
\begin{equation*}
\begin{split}
    \ket{a}\ket{f}
    \xrightarrow{\ R\ }\;
    \ket{a}\ket{f}
    -\frac{2}{\sqrt{N}}\omega^{-f(a)}\ket{\psi_f}\ket{f} 
    =\;
    \ket{a}\ket{f}
    -\frac{2}{N}\sum_{b\in [N]}\omega^{f(b)-f(a)}\ket{b}\ket{f}.
\end{split}
\end{equation*}
The phase kickback acts on the entire system. Therefore, viewing the
oracle register in the Fourier basis, we have
\begin{equation}
\label{eq:overview-fourier-query}
    \ket{a}\ket{\widehat{P}}
    \xrightarrow{\ R\ }
    \ket{a}\ket{\widehat{P}}
    -\frac{2}{N}\sum_{b\in [N]}
    \ket{b}\ket{\widehat{P-e_a+e_b}},
\end{equation}
where $\ket{\widehat{P}}$, with $P\in\mathbb{Z}_M^N$, is a Fourier basis
state, and $e_a,e_b$ are the indicator vectors for $a,b$.

For each Fourier label $P$, we view each of its $N$ entries as a signed
integer in $\{-M/2,\ldots,M/2-1\}$. We can then view $P$ as a
multiplicity-recording table for a multiset, with each entry recording a
multiplicity. Specifically, we consider two sets of elements for the
multiset: positive elements $0,1,\ldots,N-1$ and negative elements
$\overline{0},\overline{1},\ldots,\overline{N-1}$. The elements of an opposing
pair cancel each other, i.e., $\{a,\overline{a}\}=\varnothing$.
$|P(a)|$ stands for the multiplicity of element $a$ when $P(a)>0$ and the multiplicity of element $\overline{a}$ when $P(a)<0$.
In this way, we can also define addition as union, and subtraction as
union with opposite elements. Let $P'$ denote the multiset corresponding
to $P$. Then $P-e_a+e_b$ represents $P'\cup\{\overline{a},b\}$.
If $P'$ contains $a$ or $\overline{b}$, we also need to cancel out the
opposite pairs and obtain a unique minimal multiset.

Notice that, from $P'$ to $P'\cup\{\overline{a},b\}$, the number of
positive elements minus the number of negative elements remains unchanged
(we call this the \textbf{balance of the multiset}), and the size of the multiset
either changes by $-2$ or $+2$, or remains unchanged. Therefore, as in the
original compressed oracle technique, each query only makes local changes,
and we only need to focus on the non-zero entries in the Fourier label,
which are already represented by elements in the multiset.

We further track the entire execution of the algorithm. If $k=0$, the algorithm is uniform. In this case, the initial state is
$\ket{0}\ket{\widehat{0}}$, where the Fourier label is an all-zero table,
representing $\varnothing$. When $k>0$, the initial state becomes
\begin{equation}
    \frac{1}{\sqrt{M^N}}
    \sum_f\ket{\psi_f}^{\otimes k}\ket{f}
    =
    \frac{1}{\sqrt{N^k}}
    \sum_{\vec{x}\in[N]^k}
    \ket{\vec{x}}\ket{\widehat{\mu(\vec{x})}},
\label{eq:overview-initial-state}
\end{equation}
where $\mu(\vec{x})$ is the Fourier label representing
$\{x_1,x_2,\ldots,x_k\}$. Thus, the initial oracle register is in a
superposition of size-$k$, all-positive multisets. Later, as the algorithm
runs, by our observation about each query and the fact that local unitaries
do not affect the oracle register, the state is always supported on
multisets with balance $k$, and every multiset in the support has at most
$T$ negative elements after $T$ queries.

\paragraph{Define progress.}
To proceed further with the compressed oracle framework, we need to define
a proper projector $\Gamma$ so that we can use the weight on it as a
measure of the algorithm's progress. The projector $\Gamma$ should have
two properties:
\begin{enumerate}
    \item $\Gamma$ commutes with any local unitary.
    \item $\Gamma$ should be close to the winning subspace $\Pi_{k+1}$.
\end{enumerate}
The first property ensures that only queries can affect the progress and
essentially restricts $\Gamma$ to act only on the oracle register. The
second property ensures that it can be a good approximation to winning.

We define $\Gamma$ to project onto states whose supporting multisets have
at least one negative element.
The intuition comes from comparing the initial state and the final winning
state. As \Cref{eq:overview-initial-state} shows, the initial
state is supported on multisets without negative elements, and the winning
subspace is defined by
\begin{equation*}
\begin{split}
    \Pi_{k+1}
    =
    \sum_f
    \bigl(\ket{\psi_f}\bra{\psi_f}\bigr)^{\otimes k+1}
    \otimes\ket{f}\bra{f} =
    \frac{1}{N^{k+1}}
    \sum_{\vec{x},\vec{x}'\in[N]^{k+1}}\sum_P
    \ket{\vec{x}}\bra{\vec{x}'}
    \otimes
    \ket{\widehat{P+\mu(\vec{x})}}
    \bra{\widehat{P+\mu(\vec{x}')}}.
\end{split}
\end{equation*}
The eigenstates are $\{\ket{\operatorname{eig}(P)}\}_P$, with
\begin{equation*}
    \ket{\operatorname{eig}(P)}
    =
    \frac{1}{\sqrt{N^{k+1}}}
    \sum_{\vec{x}\in[N]^{k+1}}
    \ket{\vec{x}}\ket{\widehat{P+\mu(\vec{x})}}.
\end{equation*}
Since the state evolves in the subspace of multisets with balance $k$, we
only need to focus on eigenstates $\ket{\operatorname{eig}(P)}$ with $P$
of balance $-1$, which means that $P$ has at least one negative element.
Further, observe that the eigenstate is a uniform superposition over all
length-$(k+1)$ vectors. Therefore, most vectors $\vec{x}$ in the support
cannot cancel out the negative elements in $P$. Specifically, a vector
$\vec{x}$ that cancels them out has one of its entries fixed, so such
vectors account for at most a $(k+1)/N$ proportion.

We follow this intuition and proceed with the projector $\Gamma$ based on
the number of negative elements in the multiset. We can then formally show
that for a state supported on multisets with balance $k$ and without negative elements, its amplitude onto the subspace of $\Pi_{k+1}$ is at most $\sqrt{(k+1)/N}$, making $\Gamma$
a sufficient measure of winning with a $\sqrt{(k+1)/N}$ loss. The details are in the proof of \Cref{lem:final-projection-overlap}.

\paragraph{Bound the increment.}
Now we only need to bound the weight $\bigl\|\Gamma\ket{\Phi_T}\bigr\|$,
where $\ket{\Phi_T}$ is the final state after $T$ queries.
We decompose the projected state $\Gamma\ket{\Phi_T}$ into $T$ parts, with the $t$-th part corresponding to the history that moves from
$I-\Gamma$ to $\Gamma$ for the first time after the $t$-th query.
We can then bound $\bigl\|\Gamma\ket{\Phi_T}\bigr\|$ by summing the
increments in each of the $T$ queries.

To bound the increment of weight on $\Gamma$ during a single query, we
recall the view in \Cref{eq:overview-fourier-query}. This weight refers to the
case when the query inserts the first negative element into the multiset. Since the pre-query state is not in the subspace of $\Gamma$, we focus on basis states $\ket{a}\ket{\widehat{P}}$, where $P'$ does
not contain negative elements and $a\notin P'$. If the pre-query state is
supported on only one basis state, then the $2/N$ factor already bounds the
increment by $2/\sqrt{N}$.

When the pre-query state is a superposition of several basis states, the
possible convolution may cause an accumulation of amplitudes. However, in
our scenario, we have an additional orthogonality property. The previous
requirements on $a$ guarantee that $\overline{a}$ is the only negative
element in $P'\cup\{\overline{a}\}$. Therefore, the query
preserves orthogonality among the relevant basis states, and the $2/\sqrt{N}$ bound holds. A detailed proof
can be found in \Cref{lem:each-query-progress}, taking the special case
$\ell=1$.

\paragraph{Generating more copies.}
So far, we have adapted the compressed oracle framework to our problem
for the $k\to k+1$ case. When it comes to the $k\to k+r$ case with
arbitrary $r$, we still follow the framework above: purify the oracle,
define the progress projector, and bound the increment per query. When it
comes to the details, the main difference is that we need a set of
progress projectors instead of one.

If we stick to having one projector, a first thought may be
$\Gamma_{\geq r}$, projecting onto states whose multisets have at least
$r$ negative elements. This follows the intuition that the initial state
has no negative elements in its multisets, while in the eigenstates of
the winning subspace $\Pi_{k+r}$, many multisets in the support have $r$
negative elements. Such a projector is indeed a good approximation to
$\Pi_{k+r}$; however, its increment bound is too loose. With each increment
as large as $\sqrt{r/N}$, we can only obtain an $\Omega(\sqrt{N/r})$ bound
on the number of queries, which is even worse than the $k\to k+1$ bound.

The key reason that a single projector $\Gamma_{\geq r}$ is not good
enough is that it takes $r$ steps to climb from $0$ negative elements to
$r$ negative elements, while $\Gamma_{\geq r}$ can only track the last
step. It therefore inherently assumes that the state can come to the stage of having a multiset
with $r-1$ negative elements without any effort, which causes a huge loss.
The same issue occurs for any single projector.

To handle this issue and track the growth in the number of negative
elements more precisely, we define a set of projectors
$\Gamma_0,\Gamma_1,\ldots,\Gamma_{r-1},\Gamma_{\geq r}$, setting $r+1$
different stages for multisets: having $0$ negative elements, $1$ negative element, $\ldots$, $r-1$ negative elements, and no less than $r$ negative elements. With a single query, each multiset in the support of the state can only climb up one stage, and there is an exponential loss in every stage-climbing step. Therefore, after all the oracle queries, these projectors partition the state into a pyramid pattern, where we can bound the weight on each stage.

Finally, when we convert the weights on these projectors to the weight on the winning subspace $\Pi_{k+r}$, the conversion rate also differs for each stage. Intuitively, $\Gamma_{\geq r}$ converts to $\Pi_{k+r}$ at a rate of almost $1$, while $\Gamma_0$ converts at a rate of around $N^{-r}$, and there is a smooth decay in the conversion rate when the stage descends. Subsequently, with the joint effect of weight loss from climbing up stages and conversion loss, we obtain a tight bound on the winning-subspace weight, with exponential decay in $r$. More details can be found in the proof of \Cref{thm:discrete-k-to-k-plus-r}.

\subsubsection{Algorithm}

We provide a matching upper bound for generating $k+r$ copies from $k$
copies of $\ket{\psi}$ using the controlled reflection oracle
$C_{\ket{\psi}}Z$, which costs $O(\sqrt{rN})$ queries. Our key algorithm is a doubling algorithm that produces
$2k$ copies from $k$ copies of $\ket{\psi}$ with expected query complexity $O(\sqrt{kN})$.

For $r=1$, choose a Haar-random state
$\ket{g}$ and let $P=|\braket{g}{\psi}|^2$. A direct application of
amplitude amplification produces an exact copy with constant probability using
$O(P^{-1/2})$ queries~\cite{BHMT02}. It is known that $P\sim\mathsf{Beta}(1,N-1)$, and thus $\E[P^{-1/2}]=O(\sqrt N)$.

For our doubling algorithm, we will measure the $k$ copies using the following POVM.
\[
M_k(\dd g)=\binom{N+k-1}{k}(\proj{g})^{\otimes k} \dd g
\]

where $g$ ranges over the Haar measure on $n$ qubits.  Introduced 
by Hayashi~\cite{hayashiAsymptoticEstimationTheory1998}, the POVM is shown to be optimal for
the following pure state estimation problem: given $k$ copies of an unknown pure state $\ket{\psi}$,
a procedure $M$ generates a state $\ket{g}$ that maximizes the expected fidelity 
\begin{align*}
  \E_{\ket{g}\leftarrow M(\ket{\psi}^{\otimes k})}\left[|\braket{g}{\psi}|^2\right]
\end{align*}
for every $\ket{\psi}$. The optimal fidelity is $\frac{k+1}{k+N}$. Moreover,
it can be shown that $P=|\braket{g}{\psi}|^2$ follows the beta distribution $\mathsf{Beta}(k+1,N-1)$,
from which we derive $\E[P^{-1/2}]=O(\sqrt{N/k})$ when $k<N$.

Given $\ket{g}$ from the optimal state estimation POVM, we apply amplitude amplification starting with $\ket{g}$ a total of $2k$ times.
As $\E[P^{-1/2}]=O(\sqrt{N/k})$, each amplitude amplification costs $O(\sqrt{N/k})$ expected queries, and the total expected query complexity is $O(\sqrt{kN})$.
In order to prepare $r$ copies of the state, we can repeat the doubling procedure $\log r$ times.
Since the expected query counts in successive rounds form a geometric series, the total expected query complexity is $O(\sqrt{rN})$.
In~\Cref{sec:algorithms}, we also give a direct fixed-query construction achieving the same bound.

\subsubsection{Cryptographic Applications}
We propose two applications of our cloning lower bound in unclonable cryptography.

\paragraph{Quantum money.}
The first application revisits the quantum money construction discussed by Aaronson and Christiano~\cite{aaronsonQuantumMoneyHidden2012}.
They proved a query lower bound of $\Omega(\sqrt{N/k})$ for generating $k+1$ copies of a money state from $k$ copies in their complexity-theoretic no-cloning theorem.
For the same construction, we improve their lower bound to $\Omega(\sqrt{N})$ and generalize it to the $k\to k+r$ security setting with tight bounds.

\paragraph{Copy protection.}
Our second application adapts the approach of Aaronson, Liu, Liu, Zhandry,
and Zhang~\cite{aaronsonNewApproachesQuantum2021} to Haar-random program
states. Our program contains a copy of $\ket{\psi}$, and a public evaluation
oracle $\cO_{\ket{\psi},f}$ as follows:
\[
\cO_{\ket{\psi},f}=(I-\proj{\psi})\otimes I+ \proj{\psi}\otimes {O_f},
\]
where $O_f$ is the standard basis oracle for $f$. The security of the copy-protection
scheme forms a dichotomy: a pirate producing
$k+1$ useful programs from $k$ copies must either learn enough about $f$
to evaluate it without the oracle, or successfully clone $k+1$ copies of the Haar-random state
$\ket{\psi}$ from $k$ copies by querying the reflection oracle.

To make this intuition precise, we characterize the goodness of a program through the threshold implementation
introduced in~\cite{aaronsonNewApproachesQuantum2021}.
 Let $M$ be the average correctness
operator of a program circuit $C_{f,x}$, defined as
\[
  M=\E_{x}[C_{f,x}^\dagger\Pi_{f(x)}C_{f,x}],
\] 
with eigenstates $\ket{v_j}$ and eigenvalues $\lambda_j$. 
For a general state $\ket{u}$, we can write it as a superposition of eigenstates $\ket{u}=\sum_j \alpha_j\ket{v_j}$.
The threshold implementation of $M$ is a projection that selects the
eigenstates whose correctness is at least $\gamma$:
\[
  \TI_{\gamma}(M)\ket{u}
  =\sum_{j:\lambda_j\geq\gamma}\alpha_j\ket{v_j}.
\]
Conditioned on acceptance and after normalization, every state
remaining in this subspace is supported on eigenspaces with eigenvalues at least $\gamma$.
The threshold implementation provides stronger guarantees than simply $\bra{u}M\ket{u}\geq \gamma$,
as it selects the eigenstates with high correctness and is useful for extracting the good components in
$\ket{u}$. 

A copy protection scheme is $\gamma$-anti-piracy secure in the $k\to k+1$ model
 if, for any adversary $\Adv$, its output state $\sigma$ has negligible weight in the subspace
 $\otimes_{i=1}^{k+1}\TI_{\gamma}(M)$. 
 
 Ideally, if we can implement the threshold operator $\TI_{\gamma}(M)$ efficiently by querying the oracle $\cO_{\ket{\psi},f}$,
 we would have the following security dichotomy for the threshold implementation $\TI_{\gamma}(M^{(i)})$ of $\Adv$'s output program circuit:
 \begin{itemize}
  \item Either there is an $i$ such that $(\sigma_i,\TI_{\gamma}(M^{(i)}))$ has small query weight on the oracle $\cO_{\ket{\psi},f}$ on input $\ket{\psi}$.
  \item Or every threshold implementation has large query weight on the oracle $\cO_{\ket{\psi},f}$ with input $\ket{\psi}$.
 \end{itemize}
We will try to construct a learner with bounded query complexity for $f$ in the first case, and a $k\to k+1$ cloner for $\ket{\psi}$ in the second case.

A problem with the threshold implementation is that it is not efficiently implementable in general, 
but as shown in~\cite{aaronsonNewApproachesQuantum2021}, it can be approximately implemented with a circuit.
The circuit itself is no longer a projector on the state space, but we can always extend it to a projector on a larger space with ancillas initialized to $\ket{0}$ by dilation.
We use $\widehat{M}=V^\dagger \Pi V$ to denote the dilated operator.

Let $Q$ measure the average weight on $\ket{\psi}$ at a randomly chosen
query to $\cO_{\ket{\psi},f}$ made by the circuit $V$.
Let $T$ be the number of evaluation queries made by $V$, and let $V^{<t}$ be the
circuit that runs $V$ up to just before query $t$. Let $\Pi_\psi$ project its oracle's
state register onto $\ket\psi$. We can express $Q$ as
\[
  Q=\E_x\left[\frac1T\sum_{t=1}^{T}
   V^{<t\dagger}(\proj{\psi}\otimes I) V^{<t}\right].
\]
For the rest of the discussion, we will split the successful pirate's space $\Pi_{\cA}=\otimes_{i=1}^{k+1}\widehat{M}^{(i)}$
into subspaces based on the threshold implementation of $Q_{\cA}=\widehat{M}Q\widehat{M}$, i.e., the query weight restricted to
$\gamma$ good states.
\begin{align*}
  \Pi_{\clone}=\bigotimes_{i=1}^{k+1}\TI_{\beta}(Q^{(i)}_{\cA}),\qquad \Pi_{\learn}=\Pi_{\cA}-\Pi_{\clone}.
\end{align*}

For states in the space $\Pi_{\learn}$, 
we decompose the subspace into orthogonal components according to the first register with low query weight.
 For each of the components, there is a register $i$ whose state $\sigma_i$ 
 that belongs to
$\widehat{M}-\TI_{\beta}(Q_{\cA})$. This implies that the state $\sigma_i$ has $\Tr[Q_{\cA}\sigma_i]\leq O(\beta)=O((\delta/T)^2)$.
  A~\cite{BBBV97}-type hybrid argument
allows us to replace the oracle calls of $V$ to $\cO_{\ket{\psi},f}$ by the identity, resulting in
a circuit $V^{\dum}$, and causing an error of at most $O(\delta)$ in the output. Thus, we can output $\widehat{M}\sigma_i\widehat{M},C_{f,x}^{\dum}$ as our learned program.

In the other case, the joint output lies in
$\bigotimes_{i=1}^{k+1}\TI_{\beta}(Q^{(i)}_{\cA})$. To explain how to extract a
copy of $\ket\psi$ from each program, first consider a single program.
A unitary $U_{\mathrm{hist}}$ coherently 
prepares the input and query-time registers for $Q$. 
Let $\Pi_{\mathrm{init}}$ project all ancillary registers onto
their initial zero state. Orthogonality of the input and time labels gives
\[
  A:=\Pi_\psi U_{\mathrm{hist}}\Pi_{\mathrm{init}},
  \qquad A^{\dagger}A=\widehat{M}Q\widehat{M}\otimes\proj0.
\]
In particular, if $Q\ket{v_j}=\lambda_j\ket{v_j}$ with $\lambda_j>0$,
then $\ket{v_j}\ket0$ is a right singular vector of $A$ with singular
value $\sqrt{\lambda_j}$. Its corresponding left singular vector is
\[
  \ket{\ell_j}
  =\frac{A(\ket{v_j}\ket0)}{\sqrt{\lambda_j}}
  =\ket\psi\otimes\ket{w_j},
\]
because the image of $A$ lies in the range of $\Pi_\psi$.

Singular-vector transformation~\cite{gilyenQuantumSingularValue2019}
approximately replaces each singular value $\sqrt{\lambda_j}\geq
\sqrt\beta$ by $1$, while mapping its right singular vector to its left
singular vector. Consequently, it implements the map
\[
  \sum_{j:\lambda_j\geq\beta}a_j\ket{v_j}\ket0
  \longmapsto
  \sum_{j:\lambda_j\geq\beta}a_j\ket{\ell_j}
  =\ket\psi\otimes
    \sum_{j:\lambda_j\geq\beta}a_j\ket{w_j}
\]
up to a chosen operator-norm error $\zeta$.  Since the program acts on the eigenvectors of $Q$,
we can write a general state in $\Pi_{\clone}$ as a superposition of
eigenvectors with eigenvalues at least $\beta$. Thus, the whole state can be transformed to $k+1$ copies of $\ket\psi$
by parallel application of quantum singular vector transformation.
 
In~\Cref{sec:copy-protection}, we show that by replacing 
every threshold implementation $\TI_{\gamma}(M)$ with its dilated approximate version $\widehat{M}$, we can still obtain our
security dichotomy, with a set of carefully chosen error parameters.

\subsection{Other Related Work}

\paragraph{Collusion resistance and identical copies.}
Liu, Liu, Qian, and Zhandry~\cite{liuCollusionResistantCopyProtection2022}
gave plain-model constructions of copy protection for several
cryptographic functionalities with bounded collusion resistance. \c{C}akan and
Goyal~\cite{cakanUnclonableCryptographyUnbounded2024} subsequently achieved
unbounded collusion resistance for several functionalities in the plain
model, and for all unlearnable functionalities relative to a classical
oracle. These results allow pirates to obtain independently generated
programs for the same functionality. More recently, Ananth and
Goldin~\cite{ananthLessMoreCopy2026} and \c{C}akan, Goyal, Kitagawa,
Nishimaki, and Yamakawa~\cite{cakanMultiCopySecurity2025} developed
compilers that obtain security for identical pure-state copies from
security for independently generated states, under their respective
hypotheses. In particular, both works give plain-model quantum coins,
and Ananth and Goldin give a compiler for identical-copy security of
copy protection. 
In a related direction, Ananth and
Goldin~\cite{ananthLessMoreCopy2026} show that, under suitable stretch
and ancilla-size conditions, one-copy secure pseudorandom state
generators can be transformed into $t$-copy secure generators for any
polynomial $t$ fixed in advance.

\paragraph{Oracle techniques for quantum tasks.}
Our analysis uses the oracle-purification and Fourier-recording
viewpoint underlying Zhandry's compressed-oracle
technique~\cite{zhandry19}. Alagic, Majenz, and
Russell~\cite{alagicEfficientSimulation2020} developed efficient
stateful simulation of Haar-random states together with verification
and reflection access, including an alternative construction based on
phase states and compressed oracles. Poremba, Ragavan, and
Vaikuntanathan~\cite{porembaCloningGames2026} use purified random
functions and Fourier analysis to study binary-phase cloning games and
transfer their bounds to Haar cloning games. Their games involve
splitting encoded states among separated players who recover a
classical label with decoding-oracle access. Here, we study joint
verification of additional state copies produced by an algorithm with
reflection access.

\subsubsection*{AI Usage Statement}
We developed the entire lower-bound proof without AI assistance, including
the query lower bound for the PRS cloning task, the progress measure based
on database size, and the analysis bounding this measure.

We also developed the upper bound for $k\to k+1$ cloning without AI
assistance. For general $k\to k+r$ cloning, GPT Pro 5.6 provided a proof of
the upper bound matching our lower bound in response to a single prompt (while the earlier GPT Pro 5.5 did not succeed).
We subsequently worked through the argument in detail and rewrote the
analysis.

For the cryptographic applications, we first discussed the arguments and
wrote a proof outline of a few pages. We then supplied this outline to GPT
Pro to help identify suitable quantum singular value transformation (QSVT)
algorithms and parameter choices. This assistance helped us establish the
proof more quickly. We again worked through the argument in detail and rewrote the analysis.
\section{Preliminaries}
\label{sec:preliminaries}
We use $\calH_d$ to denote the $d$-dimensional complex Hilbert space $\C^d$.
For an orthogonal projector $\Pi$, define
\[
  C_\Pi Z:=(\Id-\Pi)\otimes\Id+\Pi\otimes Z.
\]

\subsection{POVMs with General Outcome Spaces}

\begin{definition}[POVM]
\label{def:povm}
Let $(\Omega,\mathcal F)$ be a measurable space and let $\mathcal H$ be a
finite-dimensional Hilbert space. A positive operator-valued measure, or
POVM, on $\mathcal H$ with outcome space $\Omega$ is a map
$M:\mathcal F\to \mathsf{Pos}(\mathcal H)$
such that $M(\emptyset)=0$, $M(\Omega)=\Id_{\mathcal H}$, and for every
countable collection of disjoint measurable sets $S_1,S_2,\ldots$,
$M\!\left(\bigcup_i S_i\right)=\sum_i M(S_i)$.
Measuring a state $\rho$ produces the
probability measure
\[
  \Pr[X\in S]=\Tr(M(S)\rho).
\]
\end{definition}

Equivalently, when $M$ is absolutely continuous with respect to a scalar
measure $\mu$, we write
\[
  M(\dd x)=E_x\,\dd\mu(x),
\]
meaning $M(S)=\int_S E_x\,\dd\mu(x)$. The normalization condition becomes
\[
  \int_\Omega E_x\,\dd\mu(x)=\Id_{\mathcal H}.
\]
For a pure state $\ket{\psi}$, the outcome density with respect to $\mu$ is
  $x\mapsto \bra{\psi}E_x\ket{\psi}$.

  We will also use POVMs supported on a subspace $V\subseteq\mathcal H$. If
$\Pi_V$ is the orthogonal projector onto $V$, then a POVM on $V$ is
normalized by $M(\Omega)=\Pi_V$. This is equivalent to a POVM on the full
space after adding one extra outcome with operator $\Id_{\mathcal H}-\Pi_V$.

\subsection{Beta Distributions and Haar Overlaps}
We will need some facts about beta distributions.

\begin{definition}[Beta distribution]
\label{def:beta}
For $\alpha,\beta>0$, the beta function is
\[
  B(\alpha,\beta)=
  \int_0^1 p^{\alpha-1}(1-p)^{\beta-1}\,dp
  =
  \frac{\Gamma(\alpha)\Gamma(\beta)}{\Gamma(\alpha+\beta)}.
\]
The beta distribution $\mathsf{Beta}(\alpha,\beta)$ is the probability distribution on $[0,1]$ with
probability density function
\[
  f_{\alpha,\beta}(p)=
  \frac{p^{\alpha-1}(1-p)^{\beta-1}}{B(\alpha,\beta)},
  \qquad 0\le p\le 1.
\]
\end{definition}

The following formulas for the moments of a beta distribution are standard.
For every real $t>-\alpha$,
\[
  \E[P^t]=\frac{B(\alpha+t,\beta)}{B(\alpha,\beta)}.
\]
In particular,
\[
  \E[P]=\frac{\alpha}{\alpha+\beta},
  \qquad
  \E[P^{-1}]=\frac{\alpha+\beta-1}{\alpha-1}
  \quad(\alpha>1).
\]

Now we show that for any Haar-random pure state $\ket{g}$, its fidelity with a fixed pure state follows a beta distribution.
\begin{lemma}[Haar overlap distribution]
\label{lem:haar-overlap}
Let $\ket{g}$ be a Haar-random pure state in $\calH_d$, and consider its fidelity with a fixed pure state
$\ket{\phi}\in\calH_d$. Then
\[
  |\braket{g}{\phi}|^2\sim \mathsf{Beta}(1,d-1).
\]
\end{lemma}

\begin{proof}
By unitary invariance of the Haar measure, assume $\ket{\phi}=\ket{e_1}$. 
It is known that $\ket{g}$ can be generated by taking a standard complex 
Gaussian vector $Z\in\C^d$ and normalizing it to unit 
length~\cite[Proposition 7.17]{watrousTheoryQuantumInformation2018}.
Sample $\ket{g}=Z/\|Z\|$, where $Z_1,\ldots,Z_d$ are independent standard complex
Gaussians. Then $|Z_1|^2,\ldots,|Z_d|^2$ are independent exponential random
variables, and
\[
  |\braket{g}{e_1}|^2
  =
  \frac{|Z_1|^2}{|Z_1|^2+\cdots+|Z_d|^2}.
\]
The vector of normalized exponential variables is Dirichlet with all
parameters equal to $1$, so its first coordinate has distribution
$\mathsf{Beta}(1,d-1)$.
\end{proof}

\subsection{Symmetric Subspaces}

For $m\ge 0$, the symmetric subspace
$\Sym^m(\C^d)\subseteq(\C^d)^{\otimes m}$ is the subspace fixed by all
permutations of the $m$ tensor factors. We write $\Pi_m$ for its
orthogonal projector. Explicitly,
\[
  \Pi_m=\frac{1}{m!}\sum_{\pi\in S_m} W_\pi,
\]
where $W_\pi$ permutes tensor factors according to $\pi$. Its dimension is
\[
  D_m:=\dim \Sym^m(\C^d)=\binom{d+m-1}{m}.
\]
Every tensor power $\ket{\psi}^{\otimes m}$ lies in $\Sym^m(\C^d)$, and
these tensor powers span the symmetric subspace. 

The following lemma from~\cite[Lemma 7.24]{watrousTheoryQuantumInformation2018} states that 
the Haar average of a tensor power is proportional to the projector onto the symmetric subspace.
\begin{lemma}[Haar moment on the symmetric subspace]
\label{lem:symmetric-haar-moment}
For $m\ge 0$,
\[
  \int (\proj{g})^{\otimes m} \dd g
  =
  \frac{\Pi_m}{D_m}.
\]
\end{lemma}

\subsection{Quantum Singular Value Transformation}
Quantum Singular Value Transformation (QSVT) is a powerful framework for developing quantum algorithms.
We will have multiple applications of QSVT in our paper. We will introduce the theorems we need in this 
section. For more details about QSVT, we refer readers to~\cite{gilyenQuantumSingularValue,gilyenQuantumSingularValue2019}.

Specifically, we will use the quantum singular vector transformation theorem~\cite[Theorem~26]{gilyenQuantumSingularValue2019}.

\begin{theorem}[Quantum singular-vector transformation
{\cite[Theorem~26]{gilyenQuantumSingularValue2019}}]
\label{thm:qsvt}
Let \(U\) be a unitary, let \(\Pi,\widetilde\Pi\) be orthogonal projectors, and set
\[
  A:=\widetilde\Pi U\Pi
  =L\Sigma R^\dagger
  =\sum_i\sigma_i\ket{\ell_i}\!\bra{r_i}.
\]
Here \(\ket{r_i}\in\supp(\Pi)\) are the right singular vectors and
\(\ket{\ell_i}\in\supp(\widetilde\Pi)\) are the corresponding left singular
vectors.

For every cutoff \(s>0\) and \(\delta>0\), there is a unitary circuit $\widetilde U$ using
\(O\!\left(\frac{1}{s}\log\frac{1}{\delta}\right)\)
applications of \(U,U^\dagger\) and $C_{\Pi}Z$ and $C_{\widetilde\Pi}Z$. 
On the right singular subspace
\(\operatorname{span}\{\ket{r_i}:\sigma_i\geq s\}\), its action is within operator
norm \(\delta\) of the partial isometry
\[
  W_{\geq s}=\sum_{i:\,\sigma_i\geq s}\ket{\ell_i}\!\bra{r_i}
\]
to the corresponding left singular subspace
\(\operatorname{span}\{\ket{\ell_i}:\sigma_i\geq s\}\).
\end{theorem}

An important corollary of the singular-vector transformation is 
the following fixed-point variant from~\cite[Theorem~27]{gilyenQuantumSingularValue2019}.

\begin{theorem}[Fixed-point amplitude amplification]
\label{thm:fix-aa}
Let $U$ be a unitary, let $\ket{\psi_0}$ be a unit vector, and let $\Pi$ be an
orthogonal projector. Suppose that
\[
  \Pi U\ket{\psi_0}=\sqrt{p} \ket{\psi_G},
\]
where $\ket{\psi_G}\in\mathsf{im}(\Pi)$ is a unit vector and
$p\ge p_{\min}>0$. Then, for every $\varepsilon\in(0,1)$, there is a unitary
circuit $\widetilde U$ such that
\[
  \left\lVert \widetilde U\ket{\psi_0}-\ket{\psi_G}\right\rVert
  \le \varepsilon.
\]
The circuit uses
 \( O\!\left(\frac{\log(1/\varepsilon)}{\sqrt{p_{\min}}}\right)\)
applications of $U$, $U^\dagger$, $C_\Pi Z$, $C_{\proj{\psi_0}}Z$.
\end{theorem}

\subsection{Projective and Threshold Implementation}

In the setting of quantum programs, defining what constitutes a good quantum program is important for characterizing security. For example,
by the uncertainty principle, repeated evaluations of a quantum program might irreversibly disturb the program and make it unusable for future evaluations.
Quantum entanglement can also cause local programs with the same behavior to interfere and exhibit different global behavior. Further discussion can be found in
~\cite{aaronsonNewApproachesQuantum2021}.

We introduce the threshold implementation from~\cite{zhandrySchrodingersPirateHow2020,aaronsonNewApproachesQuantum2021} to address the above issues. 
A common setting in security games of quantum cryptography involves the following process. Given a state $\rho$ from the adversary, 
the challenger samples a random string $r$ from some distribution $D$, and applies a binary projective measurement $\{P_r,I-P_r\}$ to $\rho$ representing accept/reject.
To characterize the success probability of the adversary, we introduce the following Hermitian operator pair $(P_D, I-P_D)$ as the mixture of projective measurements.
\begin{definition}[Mixture of projective measurements $P_D$]
  \label{def:mix-proj}
  For a set of binary projective measurements $\{P_r,I-P_r\}_{r\in R}$ and a distribution $D$ over $R$, the mixture of projective measurements $(P_D, I-P_D)$ is defined as 
  \[
    P_D:=\E_{r\sim D}[P_r],\, I-P_D:=\E_{r\sim D}[I-P_r].
  \]
\end{definition}
Note that the acceptance probability of the adversary state $\rho$ is exactly $\Tr[P_D\rho]$. We omit the distribution $D$ when it is clear from the context.

Given a mixture of projective measurements $(P_D, I-P_D)$, we define its projective and threshold implementations as follows.
\begin{definition}[Projective Implementation]
    \label{def:proj-imp}
    Write $P_D=\sum_p p\Pi_p$, where $\Pi_p$ projects onto the eigenspace
    of eigenvalue $p$. The projective implementation
    $\mathsf{ProjImp}(P_D)$ is a real-valued measurement that 
    measures $\{\Pi_p\}_p$ and returns $p\in[0,1]$.
    On input $\rho$,
    \[
      \Pr[p]=\Tr[\Pi_p\rho],\qquad
      \rho_p=\frac{\Pi_p\rho\Pi_p}{\Tr[\Pi_p\rho]}.
    \]
    The conditional state is defined for outcomes of positive probability.
    Subsequently accepting with probability $p$ reproduces the POVM
    $(P_D,I-P_D)$, since $\sum_p p\Pi_p=P_D$.
\end{definition}
The projective implementation filters out the eigenspaces of $P_D$ and returns the corresponding eigenvalue.

\begin{definition}[Threshold Implementation]
  \label{def:thres-imp}
  For a mixture of projective measurements $(P_D, I-P_D)$ and a threshold $\gamma$, we define its threshold implementation as a pair of projectors
  $(\TI_{\geq\gamma}(P_D), I-\TI_{\geq\gamma}(P_D))$. Explicitly,
  \[
  \TI_{\geq\gamma}(P_D)=\one{P_{D}\geq \gamma}.
  \]
  Since $P_D$ is a Hermitian operator, the indicator function $\one{P_{D}\geq \gamma}$ is well-defined by the spectral decomposition of $P_D$.
  Similar to $\TI_{\geq\gamma}(P_D)$, we can also define variants $\TI_{>\gamma}(P_D)$, $\TI_{\leq\gamma}(P_D)$, and $\TI_{<\gamma}(P_D)$. In order to simplify our
  notation and follow conventions in~\cite{aaronsonNewApproachesQuantum2021}, we will use $\TI_{\gamma}(P_D)$ to denote $\TI_{\geq\gamma}(P_D)$ in the rest of the paper.
\end{definition}
Intuitively, the threshold implementation $\TI_{\gamma}(P_D)$ will identify the subspace of states that have high acceptance probability with respect to the mixture of projective measurements $P_D$. 

However, it is infeasible to implement the threshold implementation $\TI_{\gamma}(P_D)$ efficiently in general. Nevertheless, we can approximate the functionality of
the threshold implementation $\TI_{\gamma}(P_D)$ using an efficient algorithm as shown in~\cite{aaronsonNewApproachesQuantum2021}. To introduce the properties of the approximate implementations,
we introduce the following definitions.

\begin{definition}[Shift distance]
    \label{def:shift-dist}
    For real-valued distributions $\mu,\nu$, define
    \[
      \Delta_{\mathsf{Shift}}^{\eps}(\mu,\nu)
      :=\inf\left\{\delta\geq0:\ \begin{array}{l}
        \Pr_{X\sim\mu}[X\leq x]\leq\Pr_{Y\sim\nu}[Y\leq x+\eps]+\delta,\\
        \Pr_{Y\sim\nu}[Y\leq x]\leq\Pr_{X\sim\mu}[X\leq x+\eps]+\delta
      \end{array}\quad\forall x\in\mathbb R\right\}.
    \]
    For real-valued measurements on the same system, set
    \[
      \Delta_{\mathsf{Shift}}^{\eps}(\cM,\mathcal N)
      :=\sup_{\rho}\Delta_{\mathsf{Shift}}^{\eps}(\cM(\rho),\mathcal N(\rho)),
    \]
    where $\cM(\rho)$ denotes the outcome distribution.
\end{definition}

\begin{definition}[$(\eps,\delta)$ Almost Projective]
    We say a real-valued measurement $\cM$ is $(\eps,\delta)$-almost projective if for every state $\rho$, two successive applications to the same system, with the second acting on the first's post-measurement state, give outcomes $X,Y$ satisfying
     \[
      \Pr[|X-Y|\leq\eps]\geq1-\delta.
    \]
\end{definition}

It is shown in~\cite{zhandrySchrodingersPirateHow2020} that there exists an efficient approximate projective implementation for any mixture of projective measurements $P_D$.
\begin{theorem}[Approximate Projective Implementation]
    \label{thm:approx-proj-imp}
    For $0<\eps,\delta<1$, there exists an approximate projective implementation for $P_D$
    $\API^{\eps,\delta}(P_D)$. It is a real-valued $(\eps,\delta)$-almost projective measurement with outcomes in $[0,1]$ satisfying
    \[
      \Delta_{\mathsf{Shift}}^{\eps}
      \bigl(\API^{\eps,\delta}(P_D),\mathsf{ProjImp}(P_D)\bigr)\leq\delta.
    \]
   The implementation has expected running time
    $T_{P,D}\cdot\poly(1/\eps,\log(1/\delta))$, where $T_{P,D}$ includes
    sampling $D$, constructing $P_r$, and performing its binary measurement.
\end{theorem}

Thresholding the approximate projective implementation gives us the approximate threshold implementation.

\begin{theorem}[Approximate Threshold Implementation]
  \label{thm:approx-thres-imp}
  For a mixture of projective measurements $(P_D, I-P_D)$, a threshold $\gamma$, and approximation parameters $\eps,\delta$, there is an approximate threshold
  implementation $\ATI^{\eps,\delta}_{\gamma}(P_D)$ with the following properties:
  \begin{align*}
  \Tr[\ATI_{\gamma-\eps}^{\varepsilon,\delta}(P_D)\rho]
  &\geq \Tr[\TI_{\gamma}(P_D)\rho]-\delta,\\
  \Tr[\TI_{\gamma-\eps}(P_D)\rho]
  &\geq  \Tr[\ATI_{\gamma}^{\varepsilon,\delta}(P_D)\rho]-\delta.
  \end{align*}
  For any quantum state $\rho$, let $\rho'$ be the post-measurement state conditioned on acceptance of $\ATI^{\eps,\delta}_{\gamma}(P_D)$ applied to $\rho$. Then
  \begin{align*}
    \Tr[\TI_{\gamma-2\eps}(P_D)\rho']\geq 1-2\delta.   
  \end{align*}
  Moreover, the approximate threshold implementation $\ATI^{\eps,\delta}_{\gamma}(P_D)$ has the same running time as $\API^{\eps,\delta}(P_D)$.
\end{theorem} 

We have the following multiparty generalization of the above properties, shown in~\cite[Corollary 3]{aaronsonNewApproachesQuantum2021}.
\begin{theorem}
  \label{thm:multi-ati}
  For each $i\in[q]$, let $\{P_{i,r},I-P_{i,r}\}_{r\in R_i}$ be binary
  projective measurements on $\calH_i$, with $D_i$ a distribution on $R_i$,
  and write $P_{D_i}:=\E_{r\sim D_i}[P_{i,r}]$.
  Fix $0<\eps,\delta,\gamma<1$. The local approximate threshold
  implementations have the following joint guarantees:
  \begin{itemize}
    \item Every state $\rho$ on $\calH_1\otimes\cdots\otimes\calH_q$,
    including mixed or entangled states, satisfies
    \[
      \Tr\!\left[\left(\bigotimes_{i=1}^{q}
        \ATI_{\gamma-\eps}^{\eps,\delta}(P_{D_i})\right)\rho\right]
      \geq
      \Tr\!\left[\left(\bigotimes_{i=1}^{q}
        \TI_{\gamma}(P_{D_i})\right)\rho\right]-q\delta.
    \]
    \item Apply $\bigotimes_{i=1}^{q}
    \ATI_{\gamma-\eps}^{\eps,\delta}(P_{D_i})$ to $\rho$.
    If joint acceptance has positive probability, denote the normalized
    state conditioned on that event by $\rho'$. Then
    \[
      \Tr\!\left[\left(\bigotimes_{i=1}^{q}
        \TI_{\gamma-2\eps}(P_{D_i})\right)\rho'\right]
      \geq 1-2q\delta.
    \]
  \end{itemize}
\end{theorem}
\section{Algorithms}
\label{sec:algorithms}

In this section we give a fixed-query algorithm for producing copies of the
target state. The algorithm has a deterministic upper bound on its number of
oracle queries. Conditioned on success, every output copy
is exact, since we can verify each prepared register using the reflection
oracle. 

We first introduce the algorithm that generates $k$ copies by only using oracle queries.  
The key subroutine is a doubling algorithm that generates $2k$ copies from $k$ copies.
We also extend the same
procedure to the \(k\to k+r\) setting.

\begin{theorem}[Fixed-query preparation for $0\to k$]
\label{thm:0-to-k}
Let \(N=2^n\) and \(1\le k\le N\). Given oracle access to
\(C_{\ket{\psi}}Z\), there is an algorithm that makes
\[
  O\!\left(\sqrt{kN}\right)
\]
queries and outputs \(\ket{\psi}^{\otimes k}\) with probability at least
\(2/3\). Conditioned on success, the output state is exactly
\(\ket{\psi}^{\otimes k}\).
\end{theorem}

We regard all oracle-independent operations as free. In particular, this
includes preparation of a classically described state, reflection about such
a state, and implementation of the covariant POVM below. The oracle also
implements the verification measurement
\(\{\proj{\psi},\Id-\proj{\psi}\}\) with one query: apply
\(C_{\ket{\psi}}Z\) to an ancilla initialized in \(\ket{+}\), and measure the
ancilla in the \(X\) basis.

We begin with the optimal POVM used to obtain a classical estimate of
\(\ket{\psi}\)~\cite{hayashiAsymptoticEstimationTheory1998}.

\begin{definition}[Symmetric pure-state POVM]
\label{def:symmetric-povm}
Let \(N=2^n\) and \(m\ge1\). The operator-valued measure
\[
  M_m(\dd g):=D_m(\proj{g})^{\otimes m}\dd g
\]
is a POVM on \(\Sym^m(\C^N)\) by
\Cref{lem:symmetric-haar-moment}.
\end{definition}

The output of this POVM is a classical description of a pure state
\(\ket g\). The following lemma describes its overlap with the target state $\ket{\psi}$.

\begin{lemma}
\label{lem:inner-product-estimation}
For every \(\ket{\psi}\in\calH_N\), apply \(M_m\) to
\(\ket{\psi}^{\otimes m}\), denote the outcome by \(\ket g\), and set
\[
  P:=|\braket{g}{\psi}|^2.
\]
Then we have that $P\sim\mathsf{Beta}(m+1,N-1)$.
Consequently,
\[
  \E[P]=\frac{m+1}{m+N},
  \qquad
  \E[P^{-1}]=\frac{m+N-1}{m}.
\]
\end{lemma}

\begin{proof}
On input \(\ket{\psi}^{\otimes m}\), the probability density function of the 
measurement outcome with respect to Haar
measure \(\dd g\) is
\[
  q_\psi(g)=D_m|\braket{g}{\psi}|^{2m}.
\]
By \Cref{lem:haar-overlap}, the squared overlap with a Haar-random outcome has
density \((N-1)(1-p)^{N-2}\). The POVM introduces a multiplicative factor of \(D_mp^m\), so
\(P\) has density
\[
  D_m(N-1)p^m(1-p)^{N-2}
  =
  \frac{p^m(1-p)^{N-2}}{B(m+1,N-1)}.
\]
This is the density of \(\mathsf{Beta}(m+1,N-1)\). The two moment formulas
follow from the beta-distribution identities in
\Cref{sec:preliminaries}.
\end{proof}

To apply fixed-point amplitude amplification, we require a promised lower bound on the
overlap for each invocation. The inverse-moment formula gives such a bound
with high probability.

For any \(c>1\), Markov's inequality applied to \(P^{-1}\) gives
\begin{align*}
  \Pr\!\left[P<\frac{m}{c^2(m+N-1)}\right]
  &=
  \Pr\!\left[P^{-1}>c^2\frac{m+N-1}{m}\right]\le \frac1{c^2}.
\end{align*}

We next introduce the doubling subroutine for our algorithm. 
We introduce an adjustable parameter \(c>1\) that controls the tradeoff between query complexity and success probability.
 The measurement in the first
line consumes the \(m\) input copies. Each accepted verification in the loop
produces one fresh register in the exact state \(\ket{\psi}\).

\begin{mdframed}[nobreak=true]
    \textbf{One fixed-query doubling step \(\mathsf{Double}(m,c)\)} 
    
    Given \(m\) copies of \(\ket{\psi}\), oracle access to
  \(C_{\ket{\psi}}Z\), and an adjustable parameter \(c>1\), proceed as follows:
    \begin{itemize}
      \item Measure the \(m\) copies using the symmetric POVM \(M_m\), obtaining the description of  \(\ket g\);
      \item Set \(p_{\min}\gets m/[c^2(m+N-1)]\) and \(\varepsilon\gets1/2\), and repeat the following process $12m$ times:
      \begin{itemize}
        \item Prepare a fresh register in the classically described state
        $\ket g$.
        \item  Apply fixed-point amplitude amplification from
    \Cref{thm:fix-aa} with \(U=\Id\), \(\Pi=\proj{\psi}\),
    \(p_{\min}\), and error \(\varepsilon\)
    \item Verify if the output is state $\ket{\psi}$.
      \end{itemize}
    The algorithm succeeds if we obtain at least $2m$ copies of $\ket{\psi}$.
    \end{itemize}
\end{mdframed}

\begin{lemma}
\label{lem:success-prob}
For \(m\ge1\), \(N\ge2\), and \(c>1\),
\(\mathsf{Double}(m,c)\) makes
\[
  O\!\left(c\sqrt{m(m+N)}\right)
\]
queries and outputs \(2m\) exact copies with probability at least
\[
  1-\frac1{c^2}-e^{-49m/18}.
\]
When \(m\le N\), the query bound is \(O(c\sqrt{mN})\).
\end{lemma}

\begin{proof}
We call the outcome \(\ket{g}\) good if
\[
  |\braket{g}{\psi}|^2
  \ge p_{\min}:=\frac{m}{c^2(m+N-1)}.
\]
By a direct application of Markov's inequality,
\begin{align*}
  \Pr[g\text{ is not good}]
  &=
  \Pr\!\left[P^{-1}>c^2\frac{m+N-1}{m}\right]\le \frac1{c^2}.
\end{align*}

Conditioned on the event that \(\ket g\) is good, each invocation of fixed-point
amplitude amplification produces a state \(\ket{\eta}\) satisfying
  $\|\ket{\eta}-\ket{\psi}\|\le\varepsilon$.
Therefore its verification accepts with probability
\[
  |\braket{\psi}{\eta}|^2
  \ge (1-\varepsilon^2/2)^2
  \ge 1-\varepsilon^2
  \ge\frac34.
\]
The \(12m\) attempts use fresh registers. Conditioned on \(\ket{g}\) being good, the number
\(X\) of accepted attempts therefore stochastically dominates a random
variable \(Y\sim\mathsf{Bin}(12m,3/4)\), whose mean is \(9m\). A Chernoff
bound gives
\[
  \Pr[X<2m]
  \le
  \Pr\!\left[Y<\frac29\E[Y]\right]
  \le
  \exp\!\left(-\frac{(7/9)^2\E[Y]}2\right)
  =
  e^{-49m/18}.
\]
Combining the two failure events provides the success bound.

By~\Cref{thm:fix-aa}, the algorithm uses
\[
  O\!\left(\frac1{\sqrt{p_{\min}}}\right)
  =
  O\!\left(c\sqrt{\frac{m+N-1}{m}}\right)
\]
queries, because \(\varepsilon=1/2\) is constant. Multiplying by \(12m\)
attempts and including verification gives an upper bound of
$O\!\left(c\sqrt{m(m+N)}\right)$.
If \(m\le N\), this simplifies to \(O(c\sqrt{mN})\).
\end{proof}

We now give the complete algorithm construction for~\Cref{thm:0-to-k}. Let
\[
  \ket{\Phi_N}
  =
  \frac1{\sqrt N}\sum_{x=1}^N\ket{x}_A\ket{x}_B.
\]
The squared overlap of \(\ket{\Phi_N}\) with the good subspace
\(\ket{\psi}_A\otimes\calH_B\) is exactly \(1/N\). Thus fixed-point amplitude
amplification with \(p_{\min}=1/N\), followed by verification of register
\(A\), produces an exact copy of \(\ket{\psi}\) with probability at least
\(3/4\) and uses \(O(\sqrt N)\) queries. We repeat this initialization 12 times,
 obtaining one verified copy, and declare failure if none succeeds.

Set
\[
  L:=\lceil\log_2 k\rceil,\qquad K:=2^L,
  \qquad c_j:=4(L-j)\quad (j=0,\ldots,L-1).
\]
At stage \(j\), apply \(\mathsf{Double}(2^j,c_j)\). Declare failure if
any stage fails; otherwise keep any \(k\) of the final \(K\) copies.

\begin{proof}[Proof of \Cref{thm:0-to-k}]
Each initialization attempt succeeds with probability at least \(3/4\), so
all 12 attempts fail with probability at most \(4^{-12}<1/100\). There are \(L\)
doubling stages, with input sizes \(m_j=2^j\) for \(j=0,\ldots,L-1\).
Conditioned on the success of all previous stages, the input to stage
\(j\) consists of $m_j$ exact copies of $\ket{\psi}$, so \Cref{lem:success-prob} applies.
The probability that some reached stage receives a bad outcome $\ket{g}$
is therefore at most
\[
  \sum_{j=0}^{L-1}\frac1{c_j^2}
  =\frac1{16}\sum_{s=1}^{L}\frac1{s^2}
  \leq\frac1{16}\left(1+\int_1^\infty\frac{\dd x}{x^2}\right)
  =\frac18.
\]
Moreover, with \(q=e^{-49/18}\),
\[
  \sum_{j=0}^{L-1}e^{-49\cdot2^j/18}
  =
  \sum_{j=0}^{L-1}q^{2^j}
  \le
  \sum_{j=0}^{\infty}q^{j+1}
  =
  \frac{q}{1-q}
  <\frac1{12}.
\]
Thus the total failure probability is less than
\(1/100+1/8+1/12<1/3\).

For every doubling stage, \(m\le K/2<k\le N\). By
\Cref{lem:success-prob}, the total number of queries is
\begin{align*}
  O(\sqrt N)
  +\sum_{j=0}^{L-1}O\!\left(c_j\sqrt{2^jN}\right)
  &=O(\sqrt N)+O\!\left(\sqrt{KN}
    \sum_{s=1}^{L}s\,2^{-s/2}\right)\\
  &=O\!\left(\sqrt{kN}\right),
\end{align*}
where \(s=L-j\), \(K<2k\), and
\(\sum_{s\geq1}s\,2^{-s/2}<\infty\).
All stage parameters and attempt counts are fixed in advance, independently
of measurement outcomes. 
\end{proof}

The following corollary extends the above algorithm to the \(k\to k+r\) setting.

\begin{corollary}[Fixed-query \(k\to k+r\) upper bound]
\label{cor:k-to-k-plus-r-upper}
Let \(N=2^n\geq2\), \(k,r\ge1\), and assume \(r\le N\). Given \(k\) copies of
\(\ket{\psi}\) and oracle access to \(C_{\ket{\psi}}Z\), there is a
fixed-query algorithm that outputs \(k+r\) exact copies with probability
at least \(2/3\) using
\[
  O(\sqrt{rN})
\]
queries.
\end{corollary}

\begin{proof}
First suppose that \(r<k\). Set aside \(k-r\) input copies and do not act on
them. Apply \(\mathsf{Double}(m,c)\) once to the remaining \(r\)
copies with \(c=4\). The step produces \(2r\) fresh verified copies.
Together with the \(k-r\) untouched copies, these give $(k-r)+2r=k+r$
copies. By \Cref{lem:success-prob}, the failure probability is at most
\(1/16+e^{-49r/18}<1/3\). Since \(r\le N\), the query cost is
\(O(\sqrt{rN})\).

Now suppose that \(r\ge k\). Starting with all \(k\) input copies, repeatedly
apply \(\mathsf{Double}(m,c)\) until the number of copies is at least
\(k+r\). Let
\[
  L:=
  \left\lceil\log_2\!\left(\frac{k+r}{k}\right)\right\rceil,
  \qquad
  M:=2^Lk.
\]
At stage \(j=0,\ldots,L-1\), apply
\(\mathsf{Double}(2^jk,c_j)\) with \(c_j:=4(L-j)\).
We declare failure if any stage fails.
Note that \(k+r\le M<2(k+r)\le4r\). As in the preceding proof, the sum of
estimation failure probabilities is at most \(1/8\), and the sum of
verification failure probabilities is at most
\(\sum_{j\geq0}e^{-49\cdot2^jk/18}<1/12\).
Thus all stages succeed with probability greater than \(2/3\).
Writing \(s=L-j\) and using \(2^jk=M2^{-s}\),
\Cref{lem:success-prob} bounds the total number of queries by
\begin{align*}
  \sum_{j=0}^{L-1}
  O\!\left(c_j\sqrt{2^jk(2^jk+N)}\right)
  &\leq O\!\left(\sqrt{M(M+N)}
    \sum_{s=1}^{L}s\,2^{-s/2}\right)\\
  &=O\!\left(\sqrt{M(M+N)}\right)
   =O(\sqrt{rN}),
\end{align*}
since \(M<4r\) and \(r\leq N\). 
\end{proof}
\section{Lower Bounds}

\subsection{Problem Setting}

We first restate the problem. Let $\ket{\psi}$ be an $n$-qubit quantum state drawn from the Haar measure $\mu_N$, with $N=2^n$. An algorithm $\cA$ is given $k$ copies of the state as input, along with oracle access to $C_{\ket{\psi}}Z$. The algorithm outputs a state $\rho$ on $k+r$ registers. It wins if $\rho$ passes the projection onto $\ket{\psi}^{\otimes k+r}$. The average winning probability of algorithm $\cA$ is
\begin{equation*}
    \Pr[\cA\text{ wins}]
    =
    \mathbb{E}_{\ket{\psi}\sim\mu_N}\ 
            \bra{\psi}^{\otimes k+r}
            \rho
            \ket{\psi}^{\otimes k+r}.
\end{equation*}

\paragraph{Discrete scenario.}

We first give the lower bounds for a slightly different scenario, where the algorithm has access to a weaker oracle $R_{\psi}=    I-2\ket{\psi}\bra{\psi}$ and $\ket{\psi}$ is drawn from a discrete distribution $\mathcal{D}_n$, which we define below. To see the relation between $R_\psi$ and $C_{\ket{\psi}} Z$, we have $C_{\ket{\psi}}Z    =    \ket{0}\bra{0}\otimes I    +    \ket{1}\bra{1}\otimes R_{\psi}$. In the end, we will show that the lower bounds also extend to the general case of controlled reflection queries $C_{\ket{\psi}}Z$ to a Haar-random state.

We now focus on the discrete scenario and describe the execution of the algorithm. Recall that $N=2^n$. Let $[N]=\{0,1,\cdots,N-1\}$, let $M$ be polynomial in $N$, and let $\omega=e^{2\pi i/M}$ be an $M$-th root of unity. Define $\mathcal{D}_n$ as the uniform distribution over the state family $\{\ket{\psi_f}\}_{f}$, where $f:[N]\to [M]$, and $\ket{\psi_f}$ is defined as
\begin{equation*}
    \ket{\psi_f}
    =
    \frac{1}{\sqrt{N}}
    \sum_{x\in[N]}
    \omega^{f(x)}\ket{x}.
\end{equation*}
We denote $R_f = R_{\psi_f}    =    I-2\ket{\psi_f}\bra{\psi_f}$ as the reflection oracle, and $\Pi_{k+r}^{(f)}    =    \left(        \ket{\psi_f}\bra{\psi_f}    \right)^{\otimes k+r}$ as the success projector. For a random function $f$, the algorithm $\cA$ is given $k$ copies of $\ket{\psi_f}$. After $T$ queries to $R_f$, it outputs $\rho^{(f)}$, and its winning probability is
\begin{equation*}
    \Pr[\cA\text{ wins}]    =    \mathbb{E}_{f}\ 
    \bra{\psi_f}^{\otimes k+r}
            \rho^{(f)}
            \ket{\psi_f}^{\otimes k+r}.
\end{equation*}

We then have the following lower bounds on the query complexity of generating extra copies.

\begin{theorem}[$k\to k+1$]
\label{thm:discrete-k-to-k-plus-one}
Let $f:[N]\to [M]$ be a random function, and let an algorithm $\cA$ start with $k$ copies of $\ket{\psi_f}$, make $T$ queries to $R_f$, and aim to output $k+1$ copies. Then
\begin{equation*}
    \Pr[\cA\text{ wins}]
    \le
    \frac{\left(2T+\sqrt{k+1}\right)^2}{N}.
\end{equation*}
Consequently, in order to generate one extra copy with constant probability, either $k=\Omega(N)$ or $T=\Omega(\sqrt{N})$.
\end{theorem}

\begin{theorem}[$k\to k+r$]
\label{thm:discrete-k-to-k-plus-r}
Let $f:[N]\to [M]$ be a random function, and let an algorithm $\cA$ start with $k$ copies of $\ket{\psi_f}$, make $T$ queries to $R_f$, and aim to output $k+r$ copies, where $r>0$. Then
\begin{equation*}
    \Pr[\cA\text{ wins}]
    \le
    \left[
        O\left(
            \frac{T}{\sqrt{rN}}
        \right)
        +
        \sqrt{\frac{k+r}{N}}
    \right]^{2r}.
\end{equation*}
Consequently, when $r=o(N)$, in order to generate $r$ extra copies, either $k=\Omega(N)$ or $T=\Omega(\sqrt{rN})$. If neither condition is met, the winning probability decays exponentially in $r$.
\end{theorem}

\subsection{Purification of the Oracle}

Before proving the lower bounds, we first include the sampling of the random
function $f:[N]\to [M]$ as part of the quantum process, using the oracle purification idea in \cite{zhandry19}. Let register $\mathsf{A}$ denote all of the algorithm's qubits, and let register $\mathsf{D}$ store the
function $f$ in the computational basis. We assume that all operations of the
algorithm are unitary, while any measurements may be deferred to the end by adding
ancilla registers.

After purifying the oracle, the initial state is
\begin{equation*}
    \ket{\Phi_0}
    =
    \frac{1}{\sqrt{M^N}}
    \sum_{f}
    \bigl(\ket{\psi_f}^{\otimes k}\ket{0}\bigr)_{\mathsf{A}}
    \otimes
    \ket{f}_{\mathsf{D}},
\end{equation*}
where the $k$ copies are initialized in the first $kn$ qubits of $\mathsf{A}$,
and all other registers of the algorithm are initialized to $\ket{0}$. We omit the zero-initialized ancillary qubits in what follows.

We define the following controlled operators and projectors on the joint system
$\mathsf{A}\otimes\mathsf{D}$:
\begin{align*}
    R
    &=
    \sum_{f}
    (R_f)_{\mathsf{A}}
    \otimes
    \ket{f}\bra{f}_{\mathsf{D}},\\
    \Pi_{k+r}
    &=
    \sum_{f}
    \bigl(\Pi_{k+r}^{(f)}\bigr)_{\mathsf{A}}
    \otimes
    \ket{f}\bra{f}_{\mathsf{D}}.
\end{align*}
In particular, setting $\Pi=\Pi_1$ gives $R=I-2\Pi$. Here, $R_f$ only acts non-trivially on the first $n$ qubits of
$\mathsf{A}$, while $\Pi_{k+r}^{(f)}$ only acts non-trivially on the first
$(k+r)n$ qubits, but for convenience we extend these operators to the entire register
$\mathsf{A}$ by tensoring with the identity.

For an algorithm $\cA$ making $T$ reflection queries, the computation can be
viewed as an alternating sequence of arbitrary oracle-independent unitaries and
reflection queries:
\begin{equation*}
    U_0\otimes I_{\mathsf{D}},\ R,\
    U_1\otimes I_{\mathsf{D}},\ R,\ \ldots,\
    U_{T-1}\otimes I_{\mathsf{D}},\ R,\
    U_T\otimes I_{\mathsf{D}},
\end{equation*}
where $U_0,\ldots,U_T$ act only on $\mathsf{A}$. For $1\le t\le T$, the state
after $t$ queries is
\begin{align*}
    \ket{\Phi_t}
    &=
    (U_t\otimes I_{\mathsf{D}})
    R
    (U_{t-1}\otimes I_{\mathsf{D}})
    \cdots
    R
    (U_0\otimes I_{\mathsf{D}})
    \ket{\Phi_0}\\
    &=
    \frac{1}{\sqrt{M^N}}
    \sum_{f}
    \bigl(
        U_tR_fU_{t-1}\cdots R_fU_0
        \ket{\psi_f}^{\otimes k}
    \bigr)_{\mathsf{A}}
    \otimes
    \ket{f}_{\mathsf{D}}.
\end{align*}

The average winning probability of the algorithm is therefore 
\begin{equation*}
    \Pr[\cA\text{ wins}]
    =
    \left\|
        \Pi_{k+r}\ket{\Phi_T}
    \right\|^2.
\end{equation*}

\paragraph{Fourier Basis.}

The picture can be further simplified by viewing the oracle register
$\mathsf{D}$ in the Fourier basis. For a length-$N$ vector
$P\in\mathbb{Z}_M^N$, define the corresponding Fourier basis state by
\begin{equation*}
    \ket{\widehat{P}}_{\mathsf{D}}
    =
    \frac{1}{\sqrt{M^N}}
    \sum_{f}
    \omega^{\langle f,P\rangle}
    \ket{f}_{\mathsf{D}},
\end{equation*}
where $\langle f,P\rangle
    =
    \sum_{x\in[N]}f(x)P(x)\pmod M$
is the inner product of $f$ and $P$. For $a\in[N]$, let $e_a\in\mathbb{Z}_M^N$ be its indicator vector; for a tuple
$\vec{x}=(x_1,\ldots,x_s)\in[N]^s$, define its multiplicity vector
$\mu(\vec{x})\in\mathbb{Z}_M^N$ with
    $\mu(\vec{x})(a)
    =
    \big|\{i\in[s]:x_i=a\}\big|$.

Representing the initial state in the Fourier basis of the oracle register
gives
\begin{equation}
\label{eq:initial-Fourier-state}
\begin{split}
     \ket{\Phi_0}
     &= \frac{1}{\sqrt{M^N}}
    \sum_{f}
    \bigl(I\cdot \ket{\psi_f}^{\otimes k}\bigr)
    \otimes
    \ket{f}\\
    &=
    \frac{1}{\sqrt{M^N}}\sum_f \sum_{\vec{x}\in[N]^k} \big(\bigotimes_{i=1}^k\ \ket{x_i}\langle x_i | \psi_f\rangle\big)\otimes \ket{f}\\
    &=
    \frac{1}{\sqrt{M^N}}\sum_f \sum_{\vec{x}\in[N]^k} \frac{1}{\sqrt{N^k}}\omega^{f(x_1)+\cdots+f(x_k)}\ket{\vec{x}}\otimes \ket{f}\\
    &=   
    \frac{1}{\sqrt{N^k}}
    \sum_{\vec{x}\in[N]^k}
    \ket{\vec{x}}_{\mathsf{A}}
    \otimes
    \ket{\widehat{\mu(\vec{x})}}_{\mathsf{D}}.
\end{split}    
\end{equation}

For each reflection query on a basis state
$\ket{a}_{\mathsf{A}}\ket{\widehat{P}}_{\mathsf{D}}$, the reflection query $R$ acts as
\begin{equation}
\label{eq:each-query}
    \begin{split}
        R\ket{a}\ket{\widehat{P}} &= \ket{a}\ket{\widehat{P}} - 2\Pi \ket{a}\ket{\widehat{P}} \\
        &= \ket{a}\ket{\widehat{P}} -  2\sum_{f} \frac{1}{\sqrt{M^N}}\omega^{\langle f,P\rangle} \langle \psi_f|a\rangle\cdot \ket{\psi_f}\ket{f}\\
        &= \ket{a}\ket{\widehat{P}} - \frac{2}{N}\sum_{b\in [N]}\sum_{f} \frac{1}{\sqrt{M^N}}\omega^{\langle f,P\rangle -f(a)+f(b)} \ket{b}\ket{f}\\
        &= \ket{a}\ket{\widehat{P}} - \frac{2}{N}\sum_{b\in [N]}\ket{b}\ket{\widehat{P-e_a+e_b}}
    \end{split}
\end{equation}
Similarly, the projector $\Pi_{k+r}$ acts on $\ket{\vec{x}}\ket{\widehat{P}}$ with $\vec{x}\in[N]^{k+r}$ as follows:
\begin{equation*}
    \Pi_{k+r}
    \ket{\vec{x}}
    \ket{\widehat{P}}
    =
    \frac{1}{N^{k+r}}
    \sum_{\vec{y}\in[N]^{k+r}}
    \ket{\vec{y}}
    \otimes
    \ket{\widehat{P-\mu(\vec{x})+\mu(\vec{y})}}.
\end{equation*}

\paragraph{Balance and Negativity.}

We track the evolution of the state in the Fourier basis by defining the balance
$\bal(P)$ and negative-size $\negg(P)$ of a Fourier label
$P\in\mathbb{Z}_M^N$. We interpret each entry of $P$ as a signed integer in $\{-M/2,-M/2+1,\cdots,M/2-1\}$ while keeping the same residue modulo $M$. In this way, we can view $P$ as a multiplicity-recording table of a multiset with signed elements. More specifically, a multiset could include positive elements $\{0,1,\cdots,N-1\}$ and negative elements $\{\overline{0},\overline{1},\cdots,\overline{N-1}\}$, while a pair $(a,\overline{a})$ cancels out, so both elements cannot appear at the same time. For each $a\in [N]$, if $P(a)<0$, the corresponding multiset of $P$ contains $|P(a)|$ copies of element $\overline{a}$; if $P(a)>0$, it contains $|P(a)|$ copies of element $a$.

We then define balance and negative-size as follows:
\begin{equation*}
    \bal(P)
    =
    \sum_{x\in[N]}P(x),
    \quad
    \negg(P)
    =
    \sum_{x\in[N]}\max\{-P(x),0\}.
\end{equation*}
Then, the balance of $P$ is the number of positive elements minus the number of negative elements, while the negative-size of $P$ is the number of negative elements.

Throughout the algorithm, only reflection queries change the amplitude on each Fourier basis state. For the balance of Fourier basis states, by \Cref{eq:initial-Fourier-state,eq:each-query}, the purified state is initially supported on Fourier basis states with balance $k$, and the balance remains unchanged during reflection queries. Therefore, for an algorithm starting with $k$ copies, the balance of the purified state is $k$ and is invariant. Also, for the negative-size of Fourier basis states, the purified state is initially supported on Fourier basis states with negative-size $0$, and each reflection query can change the negative-size by at most $1$. Consequently, we can set $M$ large enough that $M/2\gg k+T$ for the values of $k,T$ under consideration. Then every supported Fourier label $P$ satisfies $|P(x)|\leq k+T<M/2$ for every $x\in [N]$, so it remains a valid multiplicity-recording table without overflow. We can take $M$ arbitrarily large without affecting the analysis.

\paragraph{Progress Projectors.}

For $\ell\in\mathbb{N}$, define the progress projector $\Gamma_\ell$ on the entire
system $\mathsf{A}\otimes\mathsf{D}$ by
\begin{equation*}
    \Gamma_\ell
    =
    I_{\mathsf{A}}
    \otimes
    \ \ \sum_{\mathclap{\substack{P\in\mathbb{Z}_M^N,\\
                    \negg(P)=\ell}}}\ 
    \ket{\widehat{P}}\bra{\widehat{P}}_{\mathsf{D}},
\end{equation*}
which projects onto the subspace of Fourier basis states with negative-size $\ell$. We also define $\Gamma_{\ge \ell}$ for the subspace with negative-size $\geq\ell$, and $\Gamma_{\le \ell}$ for the subspace with negative-size $\leq\ell$. Notice that all these projectors act trivially on register $\mathsf{A}$. Therefore, the weight on each projector is invariant after any local unitary $U_{\mathsf{A}}\otimes I_{\mathsf{D}}$, making it a good progress measure.

\subsection{Supporting Lemmas}

We then prove three supporting lemmas, showing that in order to project the
final state onto the accepting subspace $\Pi_{k+r}$, the state needs to be
mainly supported on high-negative-size subspaces, and we bound the progress of
pumping a state from low negative-size to high negative-size.

\begin{lemma}
\label{lem:each-query-progress}
For $\ell\geq 1$ and any state $\ket{\Phi}_{\mathsf{A}\mathsf{D}}$,
\begin{equation*}
    \left\|
        \Gamma_\ell R\Gamma_{\ell-1}\ket{\Phi}
    \right\|^2
    \leq
    \frac{4\ell}{N}
    \left\|
        \Gamma_{\ell-1}\ket{\Phi}
    \right\|^2.
\end{equation*}
\end{lemma}

\begin{proof}
Without loss of generality, let $\Gamma_{\ell-1}\ket{\Phi}$ be a normalized
state lying in the subspace with negative-size $\ell-1$. We decompose it as
\begin{equation*}
    \Gamma_{\ell-1}\ket{\Phi}
    =
    \ \sum_{\mathclap{\substack{a,P\\ \negg(P)=\ell-1}}}
    \left(\ket{a}\ket{\phi_{a,P}}\right)_{\mathsf{A}}\ 
    \ket{\widehat{P}}_{\mathsf{D}},
\end{equation*}
where $\{\ket{a}\}_{a\in [N]}$ is the computational basis of the $n$-qubit query register,
and $\ket{\phi_{a,P}}$ is the unnormalized state in the rest of the
algorithm's qubits corresponding to $a,P$. Also, $\sum_{a,P, \negg(P)=\ell-1}
    \left\|\ket{\phi_{a,P}}\right\|^2=1$. Then, we have
\begin{align*}
    \Gamma_\ell R\Gamma_{\ell-1}\ket{\Phi}
    &=
    \Gamma_\ell\Gamma_{\ell-1}\ket{\Phi}
    -2\Gamma_\ell\Pi\Gamma_{\ell-1}\ket{\Phi} \\
    &=
    -2\Gamma_\ell\Pi\Gamma_{\ell-1}\ket{\Phi} \\
    &=
    -\frac{2}{N}
    \sum_{b\in [N]} \ \ket{b}
    \otimes
    \ \sum_{\mathclap{\substack{a,P\\ \negg(P)=\ell-1}}}\ 
    \ket{\phi_{a,P}}
    \big(
        \Gamma_\ell
        \ket{\widehat{P-e_a+e_b}}
    \big).
\end{align*}

We focus on the weight on basis states where $\negg(P)=\ell-1$ and
$\negg(P-e_a+e_b)=\ell$. This also gives
$\negg(P-e_a)=\ell$, which means that the multiset $P$ does not contain the element $a$, so $P(a)\leq 0$. Relaxing the constraint on $b$ when summing up, we have
\begin{align*}
    \left\|
        \Gamma_\ell R\Gamma_{\ell-1}\ket{\Phi}
    \right\|^2
    &=
    \frac{4}{N^2}
    \bigg\|
        \ \ \sum_{\mathclap{\substack{a,b,P\\
        \negg(P)=\ell-1,\\
        \negg(P-e_a+e_b)=\ell}}} \ 
        \ket{b}\ket{\phi_{a,P}}
        \ket{\widehat{P-e_a+e_b}}
    \bigg\|^2 \\
    &\leq
    \frac{4}{N^2}
    \sum_{b\in [N]}
    \bigg\|
        \ \ \sum_{\mathclap{\substack{a,P\\
        \negg(P)=\ell-1,\\
        \negg(P-e_a)=\ell}}}\ 
        \ket{\phi_{a,P}}
        \ket{\widehat{P-e_a+e_b}}
    \bigg\|^2 \\
    &=
    \frac{4}{N}
    \bigg\|
        \ \ \sum_{\mathclap{\substack{a,P\\
        \negg(P)=\ell-1,\\
        \negg(P-e_a)=\ell}}}\ 
        \ket{\phi_{a,P}}
        \ket{\widehat{P-e_a}}
    \bigg\|^2.
\end{align*}

For $Q\in\mathbb{Z}_M^N$, we define $\col(Q)    =    \left\{        (a,P):        P-e_a=Q,\         \negg(P)=\ell-1    \right\}$ to be the set of pairs $(a,P)$ that ``collide'' onto the same Fourier basis
state $\ket{\widehat{P-e_a}}$. Here we only need to consider $Q$ with $\negg(Q)=\ell$. For any $(a,P)\in \col(Q)$, viewing $P$ and $Q$ as multisets shows that $a$ must correspond to one of the $\ell$ negative elements in $Q$. Meanwhile, after $a$ is chosen, $P=Q+e_a$ is also fixed. Therefore, for any $Q$ with
$\negg(Q)=\ell$, we have $|\col(Q)|\leq\ell$. Then,
\begin{align*}
    \Big\|
        \Gamma_\ell R\Gamma_{\ell-1}\ket{\Phi}
    \Big\|^2
    &=
    \frac{4}{N}
    \Big\|
        \sum_{\substack{Q\\ \negg(Q)=\ell}}
        \big(
            \sum_{(a,P)\in\col(Q)}
            \ket{\phi_{a,P}}
        \big)
        \otimes\ket{\widehat{Q}}
    \Big\|^2 \\
    &=
    \frac{4}{N}
    \sum_{\substack{Q\\ \negg(Q)=\ell}}
    \left\|
        \sum_{(a,P)\in\col(Q)}
        \ket{\phi_{a,P}}
    \right\|^2 \\
    &\leq
    \frac{4}{N}
    \ \sum_{\mathclap{\substack{Q\\ \negg(Q)=\ell}}}\ 
    |\col(Q)|\cdot
    \ \sum_{\mathclap{(a,P)\in\col(Q)}}\ \ 
    \left\|\ket{\phi_{a,P}}\right\|^2 \\
    &\leq
    \frac{4}{N}\cdot\ell\cdot
    \sum_{\substack{Q\\ \negg(Q)=\ell}}
    \sum_{(a,P)\in\col(Q)}
    \left\|\ket{\phi_{a,P}}\right\|^2
    =
    \frac{4\ell}{N}.
\end{align*}
\end{proof}

\begin{lemma}
\label{lem:negativie-size-progress}
Let $\ket{\Phi_T}_{\mathsf{A}\mathsf{D}}$ be the purified state after $T$ reflection queries. Then, for $1\leq \ell\leq r$,
\begin{equation*}
    \bigl\|\Gamma_{\geq \ell}\ket{\Phi_T}\bigr\|^2
    \leq
    \binom{r}{\ell}
    \left(\frac{4eT^2}{rN}\right)^\ell.
\end{equation*}
\end{lemma}

\begin{proof}
Let $f(\ell,T)=\bigl\|\Gamma_{\geq \ell}\ket{\Phi_T}\bigr\|$. We prove the result by induction. First, notice the two corner cases:
\begin{itemize}
    \item If $\ell=0$, then $f(\ell,T)\leq 1$, since $\Gamma_{\geq 0}=I$.
    \item If $\ell>0$ and $T=0$, then $f(\ell,T)=0$, since $\ket{\Phi_0}$ is only supported on Fourier basis states with negative-size $0$.
\end{itemize}

We then reduce $f(\ell,T)$ to $f(\ell,T-1)$ and $f(\ell-1,T-1)$. Notice that $\ket{\Phi_T}=(U_T\otimes I)R\ket{\Phi_{T-1}}$, and the local unitary $U_T\otimes I$ commutes with the progress projector $\Gamma_{\geq \ell}$. Then
\begin{equation}
\label{eq:gamma-tail-recurrence-derivation}
\begin{split}
    \bigl\|\Gamma_{\geq \ell}\ket{\Phi_T}\bigr\|
    &=
    \bigl\|\Gamma_{\geq \ell}R\ket{\Phi_{T-1}}\bigr\| \\
    &\leq
    \bigl\|\Gamma_{\geq \ell}R\Gamma_{\geq \ell}\ket{\Phi_{T-1}}\bigr\|
    +
    \bigl\|\Gamma_{\geq \ell}R\Gamma_{\leq \ell-1}\ket{\Phi_{T-1}}\bigr\| \\
    &\leq
    \bigl\|\Gamma_{\geq \ell}\ket{\Phi_{T-1}}\bigr\|
    +
    \bigl\|\Gamma_\ell R\Gamma_{\ell-1}\ket{\Phi_{T-1}}\bigr\| \\
    &\leq
    \bigl\|\Gamma_{\geq \ell}\ket{\Phi_{T-1}}\bigr\|
    +
    \sqrt{\frac{4\ell}{N}}\,
    \bigl\|\Gamma_{\geq \ell-1}\ket{\Phi_{T-1}}\bigr\|,
\end{split}
\end{equation}
where the second inequality holds because one reflection query increases the negative-size of the Fourier label by at most one, and the only possible transition into the $\Gamma_{\geq\ell}$ subspace from below is from $\ell-1$ to $\ell$. The third inequality follows from \Cref{lem:each-query-progress} and the fact that $\Gamma_{\ell-1}\preceq\Gamma_{\geq \ell-1}$. Therefore,
\begin{equation}
    f(\ell,T)
    \leq
    f(\ell,T-1)
    +
    \sqrt{\frac{4\ell}{N}}\,f(\ell-1,T-1).
    \label{eq:gamma-tail-recurrence}
\end{equation}

We then prove by induction that
\begin{equation}
    f(\ell,T)
    \leq
    \binom{T}{\ell}
    \sqrt{\frac{4^\ell\ell!}{N^\ell}}.
    \label{eq:gamma-tail-path-bound}
\end{equation}
It is easy to check that this holds for $\ell=0$ and $T=0$. Assume that it holds for $f(\ell,T-1)$ and $f(\ell-1,T-1)$. Then, according to \Cref{eq:gamma-tail-recurrence},
\begin{align*}
    f(\ell,T)
    &\leq
    \binom{T-1}{\ell}
    \sqrt{\left(\frac{4}{N}\right)^\ell\ell!}
    +
    \sqrt{\frac{4\ell}{N}}\,
    \binom{T-1}{\ell-1}
    \sqrt{\left(\frac{4}{N}\right)^{\ell-1}(\ell-1)!} \\
    &=
    \left[
        \binom{T-1}{\ell}
        +
        \binom{T-1}{\ell-1}
    \right]
    \sqrt{\frac{4^\ell\ell!}{N^\ell}} =
    \binom{T}{\ell}
    \sqrt{\frac{4^\ell\ell!}{N^\ell}}.
\end{align*}

We further have the desired bound
\begin{align*}
    \bigl\|\Gamma_{\geq \ell}\ket{\Phi_T}\bigr\|^2
    &\leq
    \binom{T}{\ell}\binom{T}{\ell}\ell!
    \left(\frac{4}{N}\right)^\ell \\
    &\leq
    \binom{r}{\ell}
    \frac{\binom{T}{\ell}}{\binom{r}{\ell}}
    \left(\frac{4T}{N}\right)^\ell \\
    &\leq
    \binom{r}{\ell}
    \frac{(eT/\ell)^\ell}{(r/\ell)^\ell}
    \left(\frac{4T}{N}\right)^\ell \\
    &=
    \binom{r}{\ell}
    \left(\frac{4eT^2}{rN}\right)^\ell.
\end{align*}
\end{proof}

\begin{lemma}
\label{lem:final-projection-overlap}
Let $k,r>0$, and let $\ket{\Phi}_{\mathsf{A}\mathsf{D}}$ be a purified state supported on balance-$k$ path states. Then, for $0\leq \ell\leq r-1$,
\begin{equation*}
    \bigl\|\Pi_{k+r}\Gamma_\ell\ket{\Phi}\bigr\|^2
    \leq
    \binom{r}{\ell}
    \left(\frac{k+r}{N}\right)^{r-\ell}
    \bigl\|\Gamma_\ell\ket{\Phi}\bigr\|^2.
\end{equation*}
\end{lemma}

\begin{proof}
The analysis is similar to that of \Cref{lem:each-query-progress}. Without loss of generality, let $\Gamma_\ell\ket{\Phi}$ be normalized and supported on Fourier basis states with balance $k$ and negative-size $\ell$. We decompose it as
\begin{equation*}
    \Gamma_\ell\ket{\Phi}
    =
    \sum_{\vec{x},P}
    \bigl(\ket{\vec{x}}\ket{\phi_{\vec{x},P}}\bigr)_{\mathsf{A}}
    \ket{\widehat{P}}_{\mathsf{D}},
\end{equation*}
where we sum over $\vec{x}\in[N]^{k+r}$ and $P$ with $\negg(P)=\ell$ and $\bal(P)=k$. Here, $\{\ket{\vec{x}}\}$ is the computational basis of the $(k+r)n$ qubits on which $\Pi_{k+r}$ acts, and $\ket{\phi_{\vec{x},P}}$ is the corresponding unnormalized state. Also, $\sum_{\vec{x},P} \bigl\|\ket{\phi_{\vec{x},P}}\bigr\|^2=1$. Then, we have
\begin{align*}
    \bigl\|\Pi_{k+r}\Gamma_\ell\ket{\Phi}\bigr\|^2
    &=
    \bigg\|
        \frac{1}{N^{k+r}}
        \sum_{\vec{y}}\ket{\vec{y}}
        \otimes
        \sum_{\substack{\vec{x},P\\ \negg(P)=\ell}}
        \ket{\phi_{\vec{x},P}}
        \ket{\widehat{P-\mu(\vec{x})+\mu(\vec{y})}}
    \bigg\|^2 \\
    &=
    \frac{1}{N^{k+r}}
    \bigg\|
        \sum_{\substack{\vec{x},P\\ \negg(P)=\ell}}
        \ket{\phi_{\vec{x},P}}
        \ket{\widehat{P-\mu(\vec{x})}}
    \bigg\|^2.
\end{align*}

For $Q\in\mathbb{Z}_M^N$, define $\col(Q)= \left\{(\vec{x},P):P-\mu(\vec{x})=Q,\negg(P)=\ell,\bal(P)=k  \right\}$ to be the ``collision set'' of pairs $(\vec{x},P)$. We then bound the maximum size of $\col(Q)$.

For each Fourier label $Q$ with nonzero weight, $\bal(Q) =    \bal(P)-(k+r)=-r$, so $Q$ has at least $r$ negative elements. We pick out the first $r$ negative elements. For any $(\vec{x},P)\in\col(Q)$, $P$ is determined by the vector $\vec{x}$, and $\negg(Q+\mu(\vec{x}))=\ell<r$. Therefore, $\vec{x}$ has to cancel out at least $r-\ell$ of the $r$ negative elements we picked, which means that certain $r-\ell$ entries of $\vec{x}$ are fixed. We then count the number of choices that $\vec{x}$ can take:
\begin{itemize}
    \item There are $\binom{r}{\ell}$ ways of picking $r-\ell$ negative elements out of $r$.
    \item There are at most $\binom{k+r}{r-\ell}(r-\ell)!$ ways of assigning the $r-\ell$ elements to entries of $\vec{x}$.
    \item There are at most $N^{k+\ell}$ ways to fill the remaining entries of $\vec{x}$.
\end{itemize}
Therefore, for any Fourier label $Q$,
\begin{equation*}
    |\col(Q)|
    \leq
    \binom{r}{\ell}
    \binom{k+r}{r-\ell}
    (r-\ell)!
    N^{k+\ell}.
\end{equation*}

Returning to the bound on the state norm,
\begin{align*}
    \bigl\|\Pi_{k+r}\Gamma_\ell\ket{\Phi}\bigr\|^2
    &=
    \frac{1}{N^{k+r}}
    \sum_Q
    \bigg\|
        \sum_{(\vec{x},P)\in\col(Q)}
        \ket{\phi_{\vec{x},P}}
        \ket{\widehat{Q}}
    \bigg\|^2 \\
    &\leq
    \frac{1}{N^{k+r}}
    \sum_Q
    |\col(Q)|\cdot
    \sum_{(\vec{x},P)\in\col(Q)}
    \bigl\|\ket{\phi_{\vec{x},P}}\bigr\|^2 \\
    &\leq
    \frac{
        \binom{r}{\ell}
        \binom{k+r}{r-\ell}
        (r-\ell)!
        N^{k+\ell}
    }{N^{k+r}}
    \sum_Q
    \sum_{(\vec{x},P)\in\col(Q)}
    \bigl\|\ket{\phi_{\vec{x},P}}\bigr\|^2 \\
    &=
    \binom{r}{\ell}
    \binom{k+r}{r-\ell}
    (r-\ell)!
    \left(\frac{1}{N}\right)^{r-\ell} \\
    &\leq
    \binom{r}{\ell}
    \left(\frac{k+r}{N}\right)^{r-\ell}.
\end{align*}
\end{proof}

\subsection{Theorems}

With the three supporting lemmas, we then prove the query lower bounds for
generating extra state copies.
Since \Cref{thm:discrete-k-to-k-plus-one} is the $r=1$ special case of \Cref{thm:discrete-k-to-k-plus-r}, we only show the more general proof of \Cref{thm:discrete-k-to-k-plus-r}.

\begin{proof}[Proof of \Cref{thm:discrete-k-to-k-plus-r}]
As we illustrated before, for an algorithm starting with $k$ copies, the purified state is always supported on balance-$k$ Fourier basis states throughout the algorithm. Then,
\begin{align*}
    \bigl\|\Pi_{k+r}\ket{\Phi_T}\bigr\|
    &\leq
    \sum_{\ell=0}^{r-1}
    \bigl\|\Pi_{k+r}\Gamma_\ell\ket{\Phi_T}\bigr\|
    +
    \bigl\|\Pi_{k+r}\Gamma_{\geq r}\ket{\Phi_T}\bigr\| \\
    &\leq
    \sum_{\ell=0}^{r-1}
    \sqrt{\binom{r}{\ell}}
    \left(\frac{k+r}{N}\right)^{(r-\ell)/2}
    \bigl\|\Gamma_\ell\ket{\Phi_T}\bigr\|
    +
    \bigl\|\Gamma_{\geq r}\ket{\Phi_T}\bigr\| \\
    &\leq
    \sum_{\ell=0}^{r}
    \sqrt{\binom{r}{\ell}}
    \left(\frac{k+r}{N}\right)^{(r-\ell)/2}
    \bigl\|\Gamma_{\geq\ell}\ket{\Phi_T}\bigr\| \\
    &\leq
    \sum_{\ell=0}^{r}
    \binom{r}{\ell}
    \left(\frac{k+r}{N}\right)^{(r-\ell)/2}
    \left(\frac{4eT^2}{rN}\right)^{\ell/2} \\
    &=
    \left(
        2\sqrt{e}\sqrt{\frac{T^2}{rN}}
        +
        \sqrt{\frac{k+r}{N}}
    \right)^r,
\end{align*}
where the second inequality follows from
\Cref{lem:final-projection-overlap}; the third inequality follows from
$\Gamma_\ell\preceq\Gamma_{\geq\ell}$ and the fact that the coefficient of the
term $\ell=r$ in the summation equals $1$; and the fourth inequality follows from
\Cref{lem:negativie-size-progress}. Therefore,
\begin{equation*}
    \Pr[\mathrm{win}]
    =
    \bigl\|\Pi_{k+r}\ket{\Phi_T}\bigr\|^2
    \leq
    \left(
        O\left(\frac{T}{\sqrt{rN}}\right)
        +
        \sqrt{\frac{k+r}{N}}
    \right)^{2r}.
\end{equation*}
\end{proof}

\paragraph{Switching back to the Haar scenario.}
We finally lift these lower bounds to our original scenario, where the state is
drawn from the Haar-random distribution and the algorithm has access to a
controlled reflection oracle $C_{\ket{\psi}}Z$. We show that the same bounds hold by reducing
this setting to the discrete scenario. We present the two theorems for the $k\to k+1$ and $k\to k+r$ cases, and prove the latter, more general case.

\begin{theorem}\label{thm:haar-k-to-k-plus-one}
Let $\ket{\psi}$ be a Haar-random state, and let an algorithm $\cA$ start
with $k$ copies of $\ket{\psi}$, make $T$ queries to
$C_{\ket{\psi}}Z$, and aim to output $k+1$ copies. Then
\begin{equation*}
    \Pr[\cA\text{ wins}]
    \leq
    \frac{\bigl(2T+\sqrt{k+1}\bigr)^2}{N}.
\end{equation*}
\end{theorem}

\begin{theorem}\label{thm:haar-k-to-k-plus-r}
Let $\ket{\psi}$ be a Haar-random state, and let an algorithm $\cA$ start
with $k$ copies of $\ket{\psi}$, make $T$ queries to
$C_{\ket{\psi}}Z$, and aim to output $k+r$ copies, where $r>0$. Then
\begin{equation*}
    \Pr[\cA\text{ wins}]
    \leq
    \left[
        O\left(\frac{T}{\sqrt{rN}}\right)
        +
        \sqrt{\frac{k+r}{N}}
    \right]^{2r}.
\end{equation*}
\end{theorem}

\begin{proof}[Proof of \Cref{thm:haar-k-to-k-plus-r}]
\label{pf:haar-k-to-k-plus-r}
We first argue that all the bounds still hold when we replace $R_\psi$ with $C_{\ket{\psi}}Z$. Notice that $R_\psi=I-2\ket{\psi}\bra{\psi}$, and $C_{\ket{\psi}}Z    =    I-2\ket{1}\bra{1}\otimes\ket{\psi}\bra{\psi}$. The only place in the proof where the query-type matters is \Cref{lem:negativie-size-progress}. Since $\ket{1}\bra{1}\otimes\ket{\psi}\bra{\psi}    \preceq    I\otimes\ket{\psi}\bra{\psi}$, after replacing the latter projector with the former, the norm can only
become smaller, so the result still holds.

We then argue that the bound on the winning probability holds for a Haar-random
$\ket{\psi}$ as well. For any algorithm $\cA$, suppose it wins with
probability $\varepsilon$ on an $n$-qubit Haar-random state $\ket{\psi}$. We construct an
algorithm $\mathcal{B}$ with the same procedure as $\cA$, except that
it samples an $n$-qubit Haar-random unitary $V$ and:
\begin{itemize}
    \item applies $V$ to each of the $k$ copies before running the algorithm;
    \item applies $V^\dagger$ to each of the $k+r$ copies to be tested after running the algorithm;
    \item replaces each query $C_{\ket{\psi}}Z$ by
    $(I\otimes V)C_{\ket{\psi}}Z(I\otimes V^\dagger)=C_{V\ket{\psi}}Z$.
\end{itemize}

Then, when we run $\mathcal{B}$ on the discrete distribution
$\mathcal{D}_n$, it is essentially running $\cA$ on
$V\ket{\psi}$, which follows the Haar-random distribution because of $V$.
Thus, $\mathcal{B}$ wins with probability $\varepsilon$ on the discrete
distribution. By \Cref{thm:discrete-k-to-k-plus-r}, $ \varepsilon\leq  \Big(
        O\big(T/\sqrt{rN}\big)+\sqrt{(k+r)/N}\Big)^{2r}$. Therefore, the same bound also holds for the winning probability of $\cA$ under the
Haar-random distribution.
\end{proof}

\section{Cryptographic Applications}

Our cloning bounds have two applications in unclonable cryptography:
public-key quantum money and copy-protection for unlearnable function
families. They give tight counterfeiting query bounds for Haar quantum
money and strengthen the multi-copy anti-piracy analysis of a Haar-state
copy-protection construction. The latter reduction still has overhead
depending on the number of input copies.

\subsection{Public-Key Quantum Money}
\label{sec:quantum-money}
The reflection-oracle cloning model was first introduced by Aaronson~\cite{aaronsonQuantumCopyProtectionQuantum2009},
in the complexity-theoretic no-cloning theorem, which showed that generating $k+1$ copies of $\ket{\psi}$ from
$k$ copies requires $\Omega(\sqrt{2^n}/k\log k)$ queries to the reflection oracle $C_{\ket{\psi}}Z$. In subsequent work~\cite{aaronsonQuantumMoneyHidden2012},
Aaronson and Christiano improved the lower bound to $\Omega(\sqrt{2^n/k})$.
Our result substantially improves over the previous bound.

A direct application of the theorem gives tight cloning bounds for public-key quantum money.
The key component of a public-key quantum money scheme is a quantum money mini-scheme. It consists of a bank algorithm that generates a banknote and a verification algorithm that checks its validity.
To exhibit the strongest security obtained from our theorem, we use the following multi-copy secure variant of a mini-scheme.

\begin{definition}[Multi-copy secure mini-scheme]
  \label{def:pkmoney}
  A public-key quantum money mini-scheme is a pair of efficient quantum algorithms $(\mathsf{Bank},\mathsf{Ver})$,
  \begin{itemize}
      \item $\Bank(1^{\lambda})$ generates a pair $(s,\ket{\psi_s})$, where $s$ is the serial number and $\ket{\psi_s}$ is the banknote.
      \item $\Ver(s,\sigma)$ takes a serial number and an alleged banknote $\sigma$, and outputs $0/1$ to indicate whether the banknote is valid.
  \end{itemize}
  A valid quantum money scheme should satisfy the following two requirements.
  \begin{itemize}
      \item \textbf{Correctness.} For a valid pair $(s,\ket{\psi_s})\leftarrow\Bank(1^{\lambda})$, the verification algorithm should pass with high probability:
      \[
      \Pr_{(s,\ket{\psi_s})\leftarrow\Bank(1^{\lambda})}[\Ver(s,\ket{\psi_s})=1]\geq 1-\negl(\lambda).
      \]
      \item \textbf{Unclonability}. Given the serial number $s$ and $k$ copies of the banknote state $\ket{\psi_s}$, any QPT adversary $\Adv$ cannot generate $k+1$ valid banknotes. Formally speaking,
      \[
      \Pr\left[\forall i\in[k+1]: \Ver(s,\rho_i)=1 \colon (s,\ket{\psi_s})\leftarrow \Bank(1^{\lambda}),
      (\rho_1,\rho_2,\ldots,\rho_{k+1})\leftarrow\Adv(s,\ket{\psi_s}^{\otimes k})\right]\leq \negl(\lambda),
      \]
      where $\mathsf{R}_i$ is the $i$th register for the adversary.
  \end{itemize}
  
\end{definition}
We can lift a mini-scheme to a public-key quantum money scheme by using
a digital signature to sign the serial number of the banknote~\cite{aaronsonQuantumMoneyHidden2012}.

\begin{remark}
  In the above definition, we require the banknote state to be a pure state. For mixed-state banknotes, the unclonability requirement
is weaker. Viewing a mixed-state banknote as an ensemble of pure states, the adversary is
only required to sample a pure banknote from the ensemble rather than clone the exact pure state. It is shown in~\cite{ananthLessMoreCopy2026,bartusekNoteBoostingUncloneable2026}
that mixed-state security (i.i.d. copy security) can be lifted to pure-state security (identical-copy security)
using pseudorandom functions or pseudorandom unitaries in the plain model, or a random oracle or common Haar-random unitary oracle in the oracle model.
\end{remark}

Our lower bound theorem is a substantial improvement over the previous complexity-theoretic no-cloning theorem~\cite[Theorem 43]{aaronsonQuantumMoneyHidden2012}, 
as we have the following corollary from~\Cref{thm:haar-k-to-k-plus-one}.
\begin{corollary}[Tight complexity-theoretic no-cloning theorem]
  \label{cor:comp-no-cloning}
  Given $k$ copies of a Haar random state $\ket{\psi}$, as well as oracle access to $C_{\ket{\psi}}Z$, any quantum adversary $\Adv$ needs $\Omega(\sqrt{\eps N})$ queries
  prepare $k+1$ copies of $\ket{\psi}$ with $\eps$ probability, as long as $k=o(\eps N)$. Here the probability is taken over the Haar random state $\ket{\psi}$ as well
  as the behavior of the adversary $\Adv$.
\end{corollary} 

As discussed in~\cite[Section 10]{aaronsonQuantumMoneyHidden2012}, their hidden subspace construction
can also be adapted to a Haar-random state construction relative to a quantum oracle $U$.
The oracle $U$ consists of the following three oracles:
\begin{itemize}
    \item \textbf{Banknote generator $\mathcal{G}$.}
    On input a seed $r\in\{0,1\}^n$, $\mathcal{G}(r)$ generates a serial number $s_r\in\{0,1\}^{3n}$ and a
    corresponding Haar-random banknote state $\ket{\psi_{r}}$ on $n$ qubits. 
    \item \textbf{Serial number checker $\mathcal{H}$.}
  $\mathcal{H}(s)=1$ if the input $s=s_r$ for some $r$, and $\mathcal{H}(s)=0$ otherwise. 
    \item \textbf{Controlled reflection $\mathcal{R}$.}
    The reflection oracle first checks if the input serial number $s$ is valid by querying $\mathcal{H}$. If $s$ is valid,
    it applies the controlled reflection $C_{\ket{\psi_{r}}}Z$ on the banknote register.
\end{itemize}

The constructions of the $\Bank^U$ and $\Ver^U$ algorithms follow directly from the definitions.

\begin{corollary}
  There exists a quantum oracle $U$ and a quantum money mini-scheme $(\Bank^U,\Ver^U)$, such that any adversary that breaks unclonability with advantage $\eps$ requires $\Omega(\sqrt{\eps N})$ queries to the oracle $U$.
\end{corollary}
By~\cite[Corollary 18]{aaronsonQuantumMoneyHidden2012}, we can construct a secure public-key quantum money scheme relative to a quantum oracle $U'$ from the above mini-scheme.

\subsection{Quantum Copy Protection}
\label{sec:copy-protection}
In this section we show how to construct a copy-protection scheme with the reflection oracle. 
We follow the conventions from~\cite{aaronsonNewApproachesQuantum2021} to define copy-protection schemes.

A quantum program is a pair \(\mathsf P=(\rho,C)\), where \(\rho\) is the program state
 and \(C\) is the evaluator circuit.  On input \(x\), the
circuit takes $(x,\rho)$ and writes a classical output to a
designated register. For simplicity, we denote the circuit $C(x,\cdot)$ as an isometry $C_x$. 
 Let \(\Pi^{\acc}_{f,x}\) be the projector that checks whether
this output equals \(f(x)\). The following operator is a mixture of projective measurements:
\begin{equation}
  P_{D_f}
  :=
  \E_{x\leftarrow D_f}
  \left[C_x^\dagger\Pi^{\acc}_{f,x}C_x\right],
  \label{eq:program-effect}
\end{equation}
and \(\Tr(P_{D_f}\rho)\) is the average correctness of the program. 
In cryptographic applications, we often extend $\Pi^{\acc}_{f,x}$ to verify whether the output $y$ satisfies a predicate $R(x,y)$.
Examples include digital signature and encryption schemes. For digital signatures, we set $f(\cdot)=\mathsf{Sign}(\sk,\cdot)$, and the verification predicate
$R(x,y)=\Ver(\pk,x,y)$ checks if $y$ is a valid signature for $x$. The extension to the predicate is formally defined 
in~\cite[Definitions 15--17]{aaronsonNewApproachesQuantum2021}.

We say a program $\rho$ is tested as \(\gamma\)-good for computing $f$ with respect to $\cD_f$ if it is accepted by
the projective implementation $(\TI_{\gamma}(P_{D_f}),I-\TI_{\gamma}(P_{D_f}))$ from~\Cref{def:thres-imp}.

A key property of the functions that can be copy-protected is that they are unlearnable. We formally define unlearnability 
through a learning game, introduced in~\cite{aaronsonNewApproachesQuantum2021}. 

For a function $f$, its evaluation oracle $O_f$ is defined as
\[
 O_f\ket{x}\ket{y}=\ket{x}\ket{y\oplus f(x)}.
\]
We call a function unlearnable if no learner with bounded query access to $O_f$ can learn it. This notion is formalized through the following game.

\begin{definition}[Query learning game]
\label{def:learning}
Let $q(\lambda)$ be a polynomial in $\lambda$. For a function family $\cF_{\lambda}$ and a distribution family $\cD_{\lambda}=\{D_f\}$, the learning game
\(\Learn_{\cF,\cD}^{q}(\cL,\gamma)\) proceeds as follows.
\begin{enumerate}
  \item The challenger samples \(f\leftarrow\cF_\lambda\) uniformly at random.
  \item The learner \(\cL(1^\lambda)\) makes \(q(\lambda)\) quantum queries to \(O_f\), and outputs a program \((\rho,C)\). 
  \item The challenger accepts if the threshold projector $\TI_{\gamma}(P_{D_f})$ accepts $\rho$.
\end{enumerate}

\end{definition}

Using this learning game, we define unlearnability for a function family $\cF_{\lambda}$ as follows.
\begin{definition}
  \label{def:unlearnable}
  We call a function family $\cF_{\lambda}$ $\gamma$-unlearnable with respect to a distribution family $\cD_{\lambda}=\{D_f\}$ if, for every polynomial $q(\lambda)$,
  the success probability of any $q$-query learner $\cL$ in the learning game of~\Cref{def:learning} is negligible, that is,
  \begin{align*}
    \mu_{\cF,\cD,\gamma}(q)=\sup_{\cL}\Pr[\Learn_{\cF,\cD}^{q}(\cL,\gamma)=1]\leq \negl(\lambda).
  \end{align*}
\end{definition}

We now define the copy-protection scheme for unlearnable function families, in a similar fashion to~\cite{aaronsonNewApproachesQuantum2021}.
The scheme consists of three algorithms: a setup algorithm, a generation algorithm, and an evaluation algorithm.

\begin{definition}[Copy protection scheme]
\label{def:cp-scheme}
A copy-protection scheme for a function family $\cF_{\lambda}$ consists of the following three efficient algorithms\\
\((\Setup,\Gen,\Eval)\):
\begin{itemize}
  \item \(\Setup(1^\lambda)\to \sk\): The setup algorithm generates a secret key $\sk$ for the scheme. Only parties with $\sk$ can generate the copy protection program.
  \item \(\Gen(\sk,f)\to (\rho_f,\{C_{f,x}\}_{x\in [N]})\): On input $f\in \cF_{\lambda}$ and the secret key $\sk$, the generation algorithm outputs a copy protection program $(\rho_f,\{C_{f,x}\}_{x\in [N]})$.
  \item \(\Eval(\rho_f,\{C_{f,x}\}_{x\in [N]},x)\to y\): On input $x$, the evaluation algorithm computes $C_{f,x}\rho_f C_{f,x}^\dagger$ and measures the output register to obtain $y$.
\end{itemize}
We say the evaluation algorithm is oracle-free if it does not make any oracle queries to $O_f$.
\end{definition}
As we previously discussed, we can see that for digital signature and encryption schemes, the evaluation algorithm is oracle-free.

We require the copy-protection scheme to satisfy the correctness and anti-piracy properties.

\begin{definition}[Correctness]
  \label{def:cp-correctness}
  A copy-protection scheme is correct if for any $f\in \cF_{\lambda}$, any $\sk$ from $\Setup(1^\lambda)$, any program $(\rho_f,\{C_{f,x}\}_{x\in [N]})\leftarrow\Gen(\sk,f)$,
  on every input $x$, the evaluation algorithm $\Eval$ produces an output $z_{f,x}$ with probability at least $1-\negl(\lambda)$, and $z_{f,x}=f(x)$.
\end{definition}

The following anti-piracy game is a $k$-copy extension of the anti-piracy game in~\cite[Definition~17]{aaronsonNewApproachesQuantum2021}.
We require that no adversary with $k$ copies of the program can generate $k+1$ programs that are $\gamma$-good.

\begin{definition}[\(k\)-copy $\gamma$-anti-piracy game]
\label{def:antipiracy}
A $k$-copy anti-piracy game \(\Pir_{\cF,\cD,k,\gamma}^{\cA}(1^{\lambda})\) for a copy-protection scheme for $\cF,\cD$ is defined as follows.
The game has four phases.
\begin{enumerate}
  \item \textbf{Setup Phase:} The challenger samples $\sk\leftarrow\Setup(1^\lambda)$.
  \item \textbf{Sampling Phase:} The challenger samples a function $f\leftarrow\cF_\lambda$ uniformly at random.
  \item \textbf{Query Phase:} The adversary $\cA$ queries the challenger $k$ times, and obtains 
  $k$ copies of the program $(\rho_f,\{C_{f,x}\}_{x\in [N]})\leftarrow\Gen(\sk,f)$.
  \item \textbf{Output Phase:} The adversary $\cA$ outputs a possibly entangled state $\sigma$ and $k+1$ sets of unitaries $\left(\{C_{f,x}^{(i)}\}\right)_{i=1}^{k+1}$.
  They can be viewed as programs $(\sigma[\mathsf{R}_i],\{C_{f,x}^{(i)}\}_{x\in [N]})$ for $i=1,\ldots,k+1$.
\end{enumerate}
The game outputs $1$ if all the $k+1$ programs are $\gamma$-good for computing $f$ with respect to $\cD_f$.
\end{definition}

We propose the following construction for a quantum copy-protection scheme. We first discuss the construction in the oracle model, and use the oracle
construction as a heuristic to construct a copy-protection scheme in the plain model.

For a state $\ket{\psi}$, we define the evaluation oracle $\cO_{\ket{\psi}, f}$ by
\[
 \cO_{\ket{\psi},f}=(I-\proj{\psi})_{\mathsf{R}}\otimes I_{\mathsf{XY}}+ \proj{\psi}_{\mathsf{R}}\otimes {O_f}_{\mathsf{XY}}.
\]

Our copy-protection scheme follows the template of~\cite{aaronsonNewApproachesQuantum2021}, with the following three algorithms.
\begin{itemize}
  \item $\Setup(1^\lambda)\to\sk$: The setup algorithm samples a secret key $\sk\in\{0,1\}^{\lambda}$.
  \item $\Gen(\sk,f)$: Given the secret key $\sk$, the generation algorithm samples a Haar-random state $\ket{\psi_\sk}$ on $n=\poly(\lambda)$ qubits,
  and uses it as the program state.
  We set the circuit $\{C_{f,x}\}_{x\in [N]}$ as the oracle $\cO_{\ket{\psi}, f}$ with $x$ hardwired into the $\mathsf{X}$ register.
  \item $\Eval(\rho,\{C_{f,x}\},x)$:  The evaluation algorithm parses $C_{f,x}$ as acting on registers $\mathsf{RXY}$, 
 applies $C_{f,x}$ on $\rho_{\mathsf{R}}\otimes\ket{0}\bra{0}_{\mathsf{Y}}$, and measures the $Y$ register to obtain the output.
\end{itemize}

The current construction is in the oracle model, but sampling a Haar-random state is not efficient.
 Nevertheless, this construction provides a heuristic implementation of the copy-protection scheme in the plain model, as discussed in~\cite{aaronsonNewApproachesQuantum2021}.
We replace the Haar-random state with a pseudorandom state family, and replace the oracle $\cO_{\ket{\psi},f}$ with a VBB obfuscation of the circuit $\cO_{\ket{\psi},f}$.
As shown in~\cite{JLS18}, the phase state construction is pseudorandom even given access to the reflection oracle; thus, the heuristic security follows.
We can also lift the construction to an efficient construction in the oracle model, by appending a serial number to each program and 
using the oracles $\mathcal{G},\mathcal{H}$ in~\Cref{sec:quantum-money}. Setting the input to $\mathcal{G}$ to the secret key $\sk$, we can obtain our program state.

The correctness of the scheme is straightforward. We now argue that the scheme satisfies $k$-copy anti-piracy security. 
\begin{theorem}
    \label{thm:copy-protect}
    Given a copy-protection scheme for a function family $\cF_{\lambda}$,
    for any $\gamma>0$, gap $0<\Delta<\gamma$, and accuracy parameter $0<\delta<1$, the success probability of any $q_{\cA}$-query $k$-copy anti-piracy adversary $\cA$ with respect to a distribution family $\cD_{\lambda}$ is bounded by
    \begin{align*}
        \Pr[\Pir_{\cF,\cD,k,\gamma}^{\cA}(1^{\lambda})=1]
        &\leq \mu_{\cF,\cD,\gamma-\Delta}(q_{\text{learn}})\\
        &\quad+\HClone_{2^{\lambda},k}(q_{\text{clone}})+O((k+1)\delta),
    \end{align*}
    where $\mu_{\cF,\cD,\gamma-\Delta}(q)$ is the maximal success probability of a $q$-query learner in the learning game from~\Cref{def:learning},
    and $\HClone$ is the success probability of a $q$-query adversary to generate $k+1$ copies of a Haar-random state from $k$ copies as shown in~\Cref{thm:haar-k-to-k-plus-one}.
    We have that $q_{\text{learn}}$ and $q_{\text{clone}}$ are both of size $q_{\cA}+\poly\left(q_{\text{out}},\frac{1}{\Delta},\frac{1}{\delta},k\right)$, 
    where $q_{\text{out}}$ is the maximal number of queries made by any of the $k+1$ program circuits $\{C^{(i)}_{f,x}\}$. Moreover, if the evaluation algorithm
    is oracle-free, we have $q_{\text{learn}}=q_{\cA}$.
\end{theorem}

For an efficient adversary, take $k,q_{\text{out}}$ to be $\poly(\lambda)$, and fix a gap $0<\Delta<\gamma$ with $1/\Delta=\poly(\lambda)$.
For any constant $c>0$, set $\delta=\lambda^{-c}/(k+1)$. For any $(\gamma-\Delta)$-unlearnable function family, there exist corresponding polynomials $q_{\text{learn}}(\lambda)$ and $q_{\text{clone}}(\lambda)$ such that
$\mu_{\cF,\cD,\gamma-\Delta}(q_{\text{learn}})$ and $\HClone_{2^{\lambda},k}(q_{\text{clone}})$ are negligible, while the additive error is $O(\lambda^{-c})$.
Since this holds for every constant $c>0$ with the same gap $\Delta$, we conclude that $\Pr[\cA\text{ success}]$ is negligible. Therefore we have the following corollary on the security of our copy-protection scheme.
\begin{corollary}
    Our copy-protection scheme is $k$-copy $\gamma$-anti-piracy secure for any $(\gamma-\Delta)$-unlearnable function family $\cF_{\lambda}$, for any $\Delta=\Omega\left(\frac{1}{\poly(\lambda)}\right)$.
\end{corollary}

The rest of the section is dedicated to proving~\Cref{thm:copy-protect}. Our proof forms a dichotomy of the adversary's behavior, as $\cA$ will either behave as a learner for the function $f$ or as a cloner trying to clone
the state $\ket{\psi}$.

A $k$-copy anti-piracy adversary $\cA$ outputs programs $(\sigma[\mathsf{R}_i],\{C_{f,x}^{(i)}\}_{x\in [N]})$ for $i=1,\ldots,k+1$.
If the adversary is successful, every program is $\gamma$-good for computing $f$. Each program can achieve the goal by either querying the oracle 
$\cO_{\ket{\psi},f}$ with large weight or by containing enough information about the function $f$. We will try to clone $k+1$ copies of $\ket{\psi}$ in the former case,
and show that $\cA$ is a successful learner in the latter case.

For simplicity of notation, we will ignore the program number $i$ in the following discussion when we are considering a single program.
The following mixture of projective measurements $M_{D_f}$ for input program $\sigma[\mathsf R]$ with respect to distribution $D_f$ captures the probability that the program is correct.
\begin{equation}
  M_{D_f}=\E_{x\leftarrow D_f}\left[C_{f,x}^{\dagger}\Pi^{\acc}C_{f,x}\right],
\end{equation}
where $\Pi^{\acc}$ on the output register tests if the output $y$ satisfies the predicate $R(x,y)$ specified by $\Eval$.
For simplicity, we will denote $M_{D_f}$ as $M$.

For every successful adversary $\cA$ in the $k$-copy anti-piracy game, its success probability can be written as 
\[
p_{\text{pir}}= \Pr[\Pir_{\cF,\cD,k,\gamma}^{\cA}(1^{\lambda})=1] = \Tr\left[\bigotimes_{i=1}^{k+1}\TI_{\gamma}\left(M^{(i)}\right)\sigma\right].
\]
Our proof strategy is sketched as follows. We apply $\ATI_{\gamma-\eps}(M^{(i)})$ to the output state $\sigma$.
By~\Cref{thm:multi-ati}, we have 
\begin{align}
  &\Tr\!\left[\left(\bigotimes_{i=1}^{k+1}
        \ATI_{\gamma-\eps}^{\eps,\delta}(M^{(i)})\right)\sigma\right]
      \geq
      \Tr\!\left[\left(\bigotimes_{i=1}^{k+1}
        \TI_{\gamma}(M^{(i)})\right)\sigma\right]-(k+1)\delta,\label{eq:post-select-cA}\\
       &\Tr\!\left[\left(\bigotimes_{i=1}^{k+1}
        \TI_{\gamma-2\eps}(M^{(i)})\right)\sigma'\right]
      \geq 1-2(k+1)\delta. 
\end{align}
We can therefore extract the successful component of the adversary's state.

For the following section, let $m=k+1$ and take
\begin{align*}
  \eps=\frac{\Delta}{3},\qquad \beta=\frac{\delta^2}{4mT_{\ATI}^2}.
\end{align*}
We choose $\delta$ to be inverse polynomial, for example $\delta=\lambda^{-c}/m$ for any constant $c>0$.
Here $T_{\ATI}\geq1$ is a uniform worst-case query bound for the ATI dilations, with
$T_{\ATI}=\poly(q_{\text{out}},1/\eps,\log(1/\delta))$, using the bounded-runtime implementation from~\cite[Remark~6.4]{zhandrySchrodingersPirateHow2020}.
In particular, $2T_{\ATI}\sqrt{m\beta}=\delta$.

In order to describe our learner, we introduce the following operator as the dilation of the approximate threshold implementation.
For each approximate threshold implementation $\ATI_{\gamma-\eps}^{\eps,\delta}(M)$ on register $\mathsf{R}$, we can implement it coherently as the following projector
\[
\widehat{M}=V^\dagger\Pi V
\]
acting on register $\mathsf{RD}$, with the register $\mathsf{D}$ initialized as $\ket{0}_{\mathsf{D}}$. 
Given the dilation unitary $V$, we can write the decomposition $V=U_{T}O_{T}U_{T-1}\cdots U_1O_1U_0$, where each $O_i$ is a query to the oracle $\cO_{\ket{\psi},f}$. We denote
 the circuit before the $t$-th query of the oracle as $V^{<t}=U_{t-1}O_{t-1}\dots U_1O_1U_0$.
 We estimate the query weight of the program on the oracle $\cO_{\ket{\psi},f}$ through the following mixture of projective measurements over the register $\mathsf{RD}$.
\begin{equation}
  Q=\frac{1}{T}\sum_{t=1}^{T}V^{<t\dagger}(\proj{\psi}_{\mathsf{R}}\otimes I_{\mathsf{D}})V^{<t}.
\end{equation}

In the learning case, if some program $\sigma[\mathsf{R}_i]$ has small query weight on the oracle $\cO_{\ket{\psi},f}$, we
replace the oracle queries in the program circuit $\{C_{f,x}^{(i)}\}$ with the identity. The resulting circuit $\{C_{f,x}^{(i),\dum}\}$ is proven
to be a good program for computing $f$ with small loss. Correspondingly, we can define the mixture of projective measurements $M^{(i),\dum}$.
In the cloning case, we will show how to extract copies of $\ket{\psi}$ from states in the high-query-weight subspace.

For our analysis,
 we further define the restriction of $Q$ onto the accepted subspace of $\widehat{M}$ as
\begin{align*}
  \widetilde{Q}=\widehat{M}Q\widehat{M},
\end{align*}
For $\beta>0$, we also denote the corresponding threshold operator of $\widetilde{Q}$ as
\begin{align*}
  \Lambda_{\geq\beta}=\one{\widetilde{Q}\geq \beta}, \quad \Lambda_{<\beta}=\widehat{M}-\Lambda_{\geq\beta}.
\end{align*}
Since $\im(\widetilde{Q})\subseteq \im(\widehat{M})$, we have that $\Lambda_{<\beta}$ is a valid projector orthogonal to $\Lambda_{\geq\beta}$.

Using these operators, we define the projectors onto the cloning and learning branches:
\begin{align*}
  \Pi_{\clone}=\bigotimes_{i=1}^{k+1}\Lambda_{\geq\beta}^{(i)}, \quad \Pi_{\learn}=\bigotimes_{i=1}^{k+1}\widehat{M}^{(i)}-\Pi_{\clone}.
\end{align*}

\paragraph{The Learning Case.}
We construct our learner $\cL(1^{\lambda})$ for the function $f$ as follows.
\begin{enumerate}
  \item The learner $\cL$ randomly samples a Haar-random state $\ket{\phi}$, and runs the adversary $\cA$ with $k$ copies of $\ket{\phi}$,
and simulates the oracle $\cO_{\ket{\phi},f}$ by using the oracle $O_f$. The simulated $\cA$ returns a state $\sigma$ and programs $\left(\{C_{f,x}^{(i)}\}_{x\in[N]}\right)_{i=1}^{k+1}$.
  \item Replace the oracle $\cO_{\ket{\phi},f}$ in the program $\{C_{f,x}^{(i)}\}$ with the identity, and denote the resulting program as $\{C_{f,x}^{(i),\dum}\}$. We set the corresponding mixture of projective measurements as
\[
{M}^{(i),\dum}_{D_f}=\E_{x\leftarrow D_f}\left[{C}_{f,x}^{(i),\dum\dagger}\Pi^{\acc,(i)}{C}_{f,x}^{(i),\dum}\right],
\]
and define $\widehat{M}^{(i),\dum}$ correspondingly as the dilation of $\ATI_{\gamma-\eps}^{\eps,\delta}(M^{(i),\dum})$.
  \item With the output program $\sigma$ from $\cA$, $\cL$ sequentially applies $\ATI_{\gamma-\eps}^{\eps,\delta}(M^{(i),\dum})$ 
  to each program register $\mathsf{R}_i$. If a test accepts, $\cL$ outputs the corresponding program
  \[
    \left(\widetilde{\sigma}[\mathsf{R}_i],\{C_{f,x}^{(i),\dum}\}_{x\in[N]}\right),
  \]
  where $\widetilde{\sigma}[\mathsf{R}_i]$ is the post-measurement state of register $\mathsf{R}_i$.
  It aborts if all the approximate threshold implementations reject.
\end{enumerate}

We will show the following bound on the success probability of the learner $\cL$ in the learning game.
\begin{theorem}
  \label{thm:learning-bound}
  The output program  $(\widetilde{\sigma}_{\mathsf{R}_{i^*}},\{\widetilde{C}_{f,x}^{(i^*)}\}_{x\in [N]})$ from our learner $\cL$ is a $(\gamma-\Delta)$-good program for computing $f$ with respect to $\cD_f$, with success probability satisfying
  \[
  \Pr[\cL\text{  outputs a $(\gamma-\Delta)$-good program}]\geq\Tr[\Pi_{\learn}\sigma]-O((k+1)\delta).
  \]
  
  Let $q_{\cA}$ be the number of queries made by $\cA$, and let $q_{\text{out}}$ be the number of queries made by the output program of $\cA$. Our learner $\cL$ makes at most $q_{\text{learn}}=q_{\cA}+\poly\left(\frac{1}{\Delta},\frac{1}{\delta},q_{\text{out}},k\right)$ queries. Moreover, if the evaluation algorithm is oracle-free, we have $q_{\text{learn}}=q_{\cA}$.
\end{theorem}

To prove~\Cref{thm:learning-bound}, we first introduce several lemmata for our analysis.
The following lemma from~\cite{BBBV97} characterizes the relation between $\widehat{M}$ and $\widehat{M}^{\dum}$.
\begin{lemma}
  \label{lem:BBBV}
  For any input state $\ket{\varphi}\in\im(\widehat{M})$, if our dilation unitary $V$ makes $T$ queries to $\cO_{\ket{\phi},f}$,
  with $\widehat{M},\widehat{M}^{\dum},Q$ defined as above, we have that
  \[
 \norm{(\widehat{M}^{\dum}-\widehat{M})\ket{\varphi}}\leq 2T\sqrt{\bra{\varphi}Q\ket{\varphi}}.
  \]
\end{lemma}
\begin{proof}
  Note that our dilation unitary $V$ only makes queries to the oracle $\cO_{\ket{\phi},f}$ inside the program circuit $\{C_{f,x}\}$. 
  Thus replacing the oracle $\cO_{\ket{\phi},f}$ with the identity in the circuit $\{C_{f,x}\}$ is equivalent to replacing the oracle $\cO_{\ket{\phi},f}$ with the identity in the dilation unitary $V$.
  The result immediately follows from~\cite{BBBV97}.
\end{proof}

From the lemma we have the following corollary:
\begin{corollary}
  \label{cor:low-weight-overlap}
  For the operators $\widehat{M}^{\dum}$ and $\Lambda_{<\beta}$ defined above, we have that 
  \begin{align*}
    \norm{(I-\widehat{M}^{\dum})\Lambda_{<\beta}}\leq 2T\sqrt{\beta}.
  \end{align*}
\end{corollary}
\begin{proof}
  For any state $\ket{\varphi}$ in the support of $\Lambda_{<\beta}$, we have that $\Tr[Q\proj{\varphi}]\leq \beta$. Since by our construction of $\Lambda_{<\beta}$,
  the state $\ket{\varphi}$ is also in the support of $\widehat{M}$, we have that $(I-\widehat{M})\ket{\varphi}=0$.
  By~\Cref{lem:BBBV},
  \begin{align*}
    \norm{(I-\widehat{M}^\dum)\ket{\varphi}}=\norm{(\widehat{M}^\dum-\widehat{M})\ket{\varphi}}\leq 2T\sqrt{\beta}
  \end{align*}
\end{proof}

Now we are ready to prove~\Cref{thm:learning-bound}.
\begin{proof}
   The learner's success probability is characterized by the following expression:
   \begin{align*}
    \Tr\left[\left(I-\bigotimes_{i=1}^{k+1}\left( I- \ATI_{\gamma-\eps}^{\eps,\delta}(M^{(i),\dum})\right)\right)\sigma\right]=\Tr\left[\left(I-\bigotimes_{i=1}^{k+1}\left( I- \widehat{M}^{(i),\dum}\right)\right)\sigma\otimes \proj{0}_{\mathsf{W}}\right].
   \end{align*}
 We define the failure projector $F=\bigotimes_{i=1}^{k+1}\left( I- \widehat{M}^{(i),\dum}\right)$.
 For the projector $\Pi_{\learn}$ evaluated on the purification $\ket{u}$ of the state $\sigma'=\sigma\otimes \proj{0}_{\mathsf{W}}$, we have that
\begin{align*}
 \norm{\Pi_{\learn}\ket{u}}\leq \norm{\Pi_{\learn}F\ket{u}}+\norm{\Pi_{\learn}(I-F)\ket{u}}
 \leq \norm{\Pi_{\learn}F\ket{u}}+\norm{(I-F)\ket{u}}
\end{align*}
 
We now bound the term $\norm{\Pi_{\learn}F\ket{u}}$. We can decompose $\Pi_{\learn}$ according to the first register with low query weight as
\begin{align*}
  \Pi_{\learn}=\sum_{i=1}^{k+1}\left(\bigotimes_{j=1}^{i-1}\Lambda_{\geq\beta}^{(j)}\right)\otimes\Lambda_{<\beta}^{(i)}\otimes\left(\bigotimes_{j=i+1}^{k+1}\widehat{M}^{(j)}\right),
\end{align*}

For every register, the product of $\Pi_{\learn} F$ contains a term $\Lambda_{<\beta}^{(i)}(I-\widehat{M}^{(i),\dum})$. Since for projectors we have that $\norm{\Lambda_{<\beta}^{(i)}(I-\widehat{M}^{(i),\dum})}=\norm{(I-\widehat{M}^{(i),\dum})\Lambda_{<\beta}^{(i)}}\leq 2T_{\ATI}\sqrt{\beta}$ by~\Cref{cor:low-weight-overlap}. 
Thus we have that
\begin{align*}
    \norm{\Pi_{\learn}F\ket{u}}&=\norm{\left(\sum_{i=1}^{k+1}\left(\bigotimes_{j=1}^{i-1}\Lambda_{\geq\beta}^{(j)}\right)\otimes\Lambda_{<\beta}^{(i)}\otimes\left(\bigotimes_{j=i+1}^{k+1}\widehat{M}^{(j)}\right)\right)F\ket{u}}\\
    &\leq\sqrt{\sum_{i=1}^{k+1}\norm{\Lambda_{<\beta}^{(i)}(I-\widehat{M}^{(i),\dum})\ket{u}}^2}
    \leq 2T_{\ATI}\sqrt{(k+1)\beta},
\end{align*}
where the second line use the fact that the decomposition of $\Pi_{\learn}$ is orthogonal.
In conclusion, we have that
\begin{align*}
  \Pr[\cL\text{ accepts}]\geq \Tr[\Pi_{\learn}\sigma']-4T_{\ATI}\sqrt{(k+1)\beta},
\end{align*}
where $T_{\ATI}$ is the number of queries made by each approximate threshold implementation.

Let $p_i$ be the probability that register $i$ is the first accepted register, and let $\widetilde{\sigma}_i$ denote its conditional post-measurement state. 
Applying theorem~\Cref{thm:approx-proj-imp} to each state $\tilde{\sigma}_i$ for $\ATI_{\gamma-\eps,\delta}(M^{(i),\dum})$, and
summing over at most $k+1$ attempted tests gives
\[
  \sum_{i=1}^{k+1}p_i\Tr[(I-\TI_{\gamma-3\eps}(M^{(i),\dum}))\widetilde{\sigma}_i]\leq2(k+1)\delta.
\]
Thus we obtain that
\begin{align*}
  \Pr[\cL\text{ outputs a $(\gamma-3\eps)$-good program}]
  &\geq \Pr[\cL\text{ accepts}]-2(k+1)\delta\\
  &\geq \Tr[\Pi_{\learn}\sigma']-4T_{\ATI}\sqrt{(k+1)\beta}-2(k+1)\delta.
\end{align*}
Our choices $3\eps=\Delta$ and $4T_{\ATI}\sqrt{(k+1)\beta}=\delta$ give a loss of $O((k+1)\delta)$, proving the claimed goodness bound.
The number of queries made by $\cL$ is bounded by the following terms:
  \begin{itemize}
    \item the number of queries $q_{\cA}$ made by $\cA$.
    \item the number of queries made by $\ATI_{\gamma-\eps}^{\eps,\delta}(M^{(i),\dum})$. Note that $M^{(i),\dum}$ will only query the oracle $O_f$ for the evaluation of $f$;
    thus for general programs the at most $m$ tests cost $\poly\left(q_{\text{out}},1/\eps,\log(1/\delta),k\right)$ queries. For oracle-free evaluation algorithms, this step will not create additional queries.
  \end{itemize}  
  Summing the query counts, we have that the total number of queries made by $\cL$ is $q_{\cA}+\poly\left(q_{\text{out}}, \frac{1}{\Delta},\frac{1}{\delta},k\right)$.
\end{proof}

\paragraph{The Cloning Case.}
As we have already bounded the success probability of the learning case, we now focus on the cloning case.
We will show that for states in the image of $\Pi_{\clone}$, we are able to construct a cloning algorithm $\cC$ that generates $k+1$ copies of $\ket{\psi}$ given $k$ copies of $\ket{\psi}$ and the reflection oracle $C_{\ket{\psi}}Z$.

The cloning algorithm $\cC$ is constructed as follows.  It first simulates $\cA$ by sampling $f\leftarrow\cF_{\lambda}$ and simulates the oracle $\cO_{\ket{\psi},f}$ by using the reflection oracle.
On the output state $\sigma$, it applies quantum singular vector transformation (\Cref{thm:qsvt}) in parallel, with the parameters specified in~\Cref{thm:clone-extract} to extract $k+1$ copies of $\ket{\psi}$.

Now we show that given a state $\ket{\Phi}$ in the space $\Pi_{\clone}=\bigotimes_{i=1}^{k+1}\Lambda_{\geq\beta}^{(i)}$, we can extract $k+1$ copies of $\ket{\psi}$ with high probability using singular vector transformation.
\begin{theorem}
  \label{thm:clone-extract}
  For a quantum state $\ket{\Phi}$ in the space $\Pi_{\clone}=\bigotimes_{i=1}^{k+1}\Lambda_{\geq\beta}^{(i)}$, there exists a quantum algorithm that takes $\ket{\Phi}$ and outputs a state $\ket{\Phi'}$ such that
  \[
    \norm{\ket{\Phi'}-\ket{\psi}^{\otimes k+1}\otimes\ket{v}}\leq (k+1)\delta,
  \]
  where $\ket{v}$ is some state in the work register.
  The algorithm makes $O\left(\frac{T_{\ATI}}{\sqrt{\beta}}\log\frac{1}{\delta}\right)$ queries to the oracle $C_{\ket{\psi}}Z$ for each $i=1,\ldots,k+1$, where $T_{\ATI}$ bounds the number of oracle queries in each approximate threshold implementation dilation $\widehat{M}$.
\end{theorem}

\begin{proof}
We focus on a single register of the program first. 
We first dilate the mixture of projective measurements $Q$. We additionally introduce a time register $\mathsf{T}$.
The input state to $Q$ first passes the projective measurement $\widehat{M}_{\mathsf{RD}}$, so it is a state on register $\mathsf{RD}$.
We write the history state of an arbitrary input $\ket{\phi}_{\mathsf{RD}}$ as
\[
\ket{\phi_{\text{hist}}}=\sum_{t=1}^{T}\frac{1}{\sqrt{T}}\ket{t}_{\mathsf{T}}V^{<t}\ket{\phi}_{\mathsf{RD}}.
\]
We define the unitary $U_{\text{hist}}$ as the state preparation circuit that maps the input state $\ket{\phi}_{\mathsf{RD}}\ket{0}_{\mathsf{T}}$ to its history state $\ket{\phi_{\text{hist}}}_{\mathsf{RDT}}$.
Note that the unitary $U_{\text{hist}}$ is efficiently implementable, by first preparing the state $\sum_{t=1}^{T}\sqrt{\frac{1}{T}}\ket{t}_{\mathsf{T}}$,
and then applying the controlled prefix $V^{<t}$ on the register $\mathsf{RD}$.

Define the following projectors on the registers $\mathsf{RDT}$:
\begin{align*}
  \Pi_{\text{init}}=\widehat{M}_{\mathsf{RD}}\otimes\proj{0}_{\mathsf{T}},\quad \Pi_{\text{state}}=\proj{\psi}_{\mathsf{R}}\otimes I_{\mathsf{DT}}.
\end{align*}
By our definition of $U_{\text{hist}}$, we have the following identity:
\begin{align*}
  \Pi_{\text{init}}U_{\text{hist}}^{\dagger}\Pi_{\text{state}}U_{\text{hist}}\Pi_{\text{init}}=\widetilde{Q}_{\mathsf{RD}}\otimes\proj{0}_{\mathsf{T}}.
\end{align*}
Let $A=\Pi_{\text{state}}U_{\text{hist}}\Pi_{\text{init}}$. We state several facts about its singular value decomposition. Since $A^{\dagger}A=\widetilde{Q}_{\mathsf{RD}}\otimes\proj{0}_{\mathsf{T}}$, we can see that the non-zero singular values of $A$ are the square roots of the non-zero eigenvalues of $\widetilde{Q}$,
and the corresponding right singular vectors of $A$ are the eigenvectors of $\widetilde{Q}$ on $\mathsf{RD}$ tensored with $\ket{0}_{\mathsf{T}}$.
For the left singular vectors of $A$, every left singular vector $\ket{l_i}$ with non-zero singular value is in the space $\Pi_{\text{state}}$, and can be written as $\ket{l_i}=\ket{\psi}\otimes\ket{v_i}$.

With these two observations, we apply~\Cref{thm:qsvt} with the following parameter setting:
\[
  \widetilde{\Pi}\leftarrow\Pi_{\text{state}},\quad \Pi\leftarrow\Pi_{\text{init}},U=U_{\text{hist}},\quad s\leftarrow\sqrt{\beta},\quad \delta\leftarrow \delta.
\]
The algorithm outputs a unitary $U_{\text{qsvt}}$ that is $\delta$ close to the partial isometry $W_{\geq \sqrt{\beta}}=\sum_{i\colon\sigma_i\geq\sqrt{\beta}}\ket{l_i}\bra{r_i}$
on the right singular subspace $\operatorname{span}\{\ket{r_i}\colon\sigma_i\geq\sqrt{\beta}\}$. From our previous observation, this is exactly the image of the projector $\one{\widetilde{Q}\geq \beta}\otimes\proj{0}_{\mathsf{T}}=\Lambda_{\geq\beta}\otimes\proj{0}_{\mathsf{T}}$.
Every input in the image of $\Lambda_{\geq\beta}$, tensored with $\ket{0}_{\mathsf{T}}$, is in the right singular subspace of $W_{\geq \sqrt{\beta}}$.
For any input state $\ket{\varphi}$ in $\Lambda_{\geq \beta}$, we have that
\begin{align*}
    \norm{\ket{\psi}\otimes\ket{v}-U_{\text{qsvt}   }\ket{\varphi}\otimes{\ket{0}_{\mathsf{T}}}}\leq \delta
\end{align*}
for some state $\ket{v}$.

From our analysis above, we observe that the algorithm is valid for any state in the space $\Lambda_{\geq\beta}$.
Therefore for the state $\ket{\Phi}$ in the space $\Pi_{\clone}=\bigotimes_{j=1}^{k+1}\Lambda_{\geq\beta}^{(j)}$,
we can append zero ancillas and write it as a sum of tensor products of the corresponding local right singular vectors,
\begin{align*}
  \ket{\Phi}\bigotimes_{j=1}^{k+1}\ket{0}_{\mathsf{T}_j}=\sum_{i_1,\ldots,i_{k+1}}\alpha_{i_1,\ldots,i_{k+1}}\ket{r_{i_1}}\otimes\cdots\otimes\ket{r_{i_{k+1}}}.
\end{align*}
Since each state $\ket{r_{i_j}}$ is in the space $\Lambda_{\geq\beta}^{(j)}\otimes\proj{0}_{\mathsf{T}_j}$, we have that 
\begin{align*}
    \ket{\Phi'}=U_{\text{qsvt}}^{\otimes{k+1}}\ket{\Phi}\bigotimes_{j=1}^{k+1}\ket{0}_{\mathsf{T}_j}
\end{align*}
 is $(k+1)\delta$ close to $\ket{\psi}^{\otimes k+1}\otimes\ket{\tilde{v}}$ for some state $\ket{\tilde v}$ by an application of the triangle inequality and reordering the registers.

\end{proof}

We can therefore bound the success probability of the cloning algorithm $\cC$ as follows.
\begin{theorem}
  \label{thm:cloning-bound}
  Our cloning algorithm $\cC$ extracts $k+1$ copies of $\ket{\psi}$ with success probability satisfying
  \begin{align*}
  \Pr[\cC\text{ outputs } \ket{\psi}^{\otimes (k+1 )}]\geq \Tr[\Pi_{\clone}\sigma]-(k+1)\delta.
\end{align*}
The cloner $\cC$ makes at most $q_{\cA}+\poly\left(q_{\text{out}},\frac{1}{\Delta},\frac{1}{\delta},k\right)$ queries to the oracle $C_{\ket{\psi}}Z$.
\end{theorem}

\begin{proof}
  Let $U$ be the parallel extraction unitary from~\Cref{thm:clone-extract}. The final verification of the extracted registers has input-space acceptance projector
  \[
    \Pi_{k+1}=U^\dagger\left(\bigotimes_{i=1}^{k+1}\left(\proj{\psi}_{\mathsf{R}_i}\otimes I_{\mathsf{D}_i\mathsf{T}_i}\right)\right)U.
  \]
  Include the zero-initialized dilation registers in $\sigma$, and let $\mathsf T$ collect all time registers. Set $\sigma'=\sigma\otimes\proj{0}_{\mathsf T}$, $\Pi'_{\clone}=\Pi_{\clone}\otimes\proj{0}_{\mathsf T}$, and $a=(k+1)\delta$.
  By~\Cref{thm:clone-extract}, for any state $\ket{\varphi}$ in $\Pi_{\clone}'$,
   $\norm{(I-\Pi_{k+1})\ket{\varphi}}\leq (k+1)\delta$. For the purification $\ket{u}$ of state $\sigma'$, we set $p=\bra{u}\Pi'_{\clone}\ket{u}$ and let its normalized cloning component as $\ket{h}=\Pi'{\clone}\ket{u}/\sqrt{p}$.
  Since $\norm{(I-\Pi_{k+1})\ket{h}}\leq (k+1)\delta$, by an application of gentle measurement lemma on $\ket{h'}=\Pi_{k+1}\ket{h}/\norm{\Pi_{k+1}\ket{h}}$, we have that
  \[
    \bra{u}\Pi_{k+1}\ket{u}\geq |\braket{u} {h'}|^2\geq\abs{\braket{u}{h}}^2-(k+1)\delta
  \]
  This implies
  \begin{align*}
    \Pr[\cC\text{ succeeds}]&=\Tr[\Pi_{k+1}\sigma']\\
    &\geq\Tr[\Pi'_{\clone}\sigma']-(k+1)\delta\\
    &=\Tr[\Pi_{\clone}\sigma]-(k+1)\delta.
  \end{align*}
  For the number of queries made by $\cC$, it first simulates $\cA$ and then uses an additional $O\left(\frac{k+1}{\sqrt{\beta}}\log\frac{1}{\delta}\right)$ applications of $U_{\text{hist}},U_{\text{hist}}^{\dagger}$ and controlled reflections about $\Pi_{\text{init}},\Pi_{\text{state}}$,
  which is $O\left(\frac{(k+1)T_{\ATI}}{\sqrt{\beta}}\log\frac{1}{\delta}\right)$ queries to the oracle $C_{\ket{\psi}}Z$. 
   Thus the total number of queries made by $\cC$ is $q_{\cA}+\poly\left(q_{\text{out}},\frac{1}{\Delta},\frac{1}{\delta},k\right)$.
\end{proof}

Now we are ready to prove~\Cref{thm:copy-protect}.
\begin{proof}[Proof of~\Cref{thm:copy-protect}]
  Combining~\Cref{thm:learning-bound} and~\Cref{thm:cloning-bound}, we obtain
\begin{align*}
  \Pr[\cA\text{ success}]&\leq \Tr[\bigotimes_{i=1}^{k+1}\ATI_{\gamma-\eps}^{\epsilon,\delta}(M^{(i)})\sigma]+(k+1)\delta\\
  &=\Tr[\bigotimes_{i=1}^{k+1}\widehat{M^{(i)}}\sigma']+(k+1)\delta\\
  &= \Tr[\Pi_{\learn}\sigma']+\Tr[\Pi_{\clone}\sigma']+(k+1)\delta\\
  &\leq \mu_{\cF,\cD,\gamma-\Delta}(q_{\text{learn}})+\HClone_{2^{\lambda},k}(q_{\text{clone}})+(4k+5)\delta.
\end{align*}
The first line follows from~\Cref{eq:post-select-cA}, the second from the dilation, and the third from the definition of the projectors. The total error is $m\delta+(2m+1)\delta+m\delta=(4m+1)\delta=O((k+1)\delta)$.
\end{proof}
\iflncs
  \bibliographystyle{splncs04}
\else
  \bibliographystyle{alpha}
\fi
\bibliography{ref}
\appendix
\renewcommand{\theHsection}{\Alph{section}}

\end{document}